\documentclass[runningheads,a4paper]{llncs}

\usepackage[USenglish]{babel}

\usepackage{url}

\usepackage{amsmath}
\usepackage{amssymb}
\usepackage{url}

\usepackage[breaklinks=true]{hyperref}
\usepackage{subfigure}
\usepackage{tikz}
\usepackage{multirow}
\usepackage{microtype}
\usetikzlibrary{automata,positioning,shapes,fit,calc}
\usepackage{colonequals}
\usepackage{stmaryrd}
\usepackage{cleveref}
\usepackage{nicefrac}
\usepackage{mdframed}
\usepackage{pifont}
\usepackage{paralist}

\newcommand{\cmark}{\text{\ding{51}}}%
\newcommand{\xmark}{\text{\ding{55}}}%

\newcommand{\storm}{\textsc{Storm}}

\newcommand{\RR}{\mathbb{R}}
\newcommand{\RRnn}{\mathbb{R}_{\ge0}}
\newcommand{\NN}{\mathbb{N}}
\newcommand{\Distr}{\mathsf{Distr}}
\newcommand{\supp}{\mathsf{supp}}

\newcommand{\mdp}{\mathcal{M}}
\newcommand{\trsys}{\mathsf{KS}}
\newcommand{\mc}{\mathsf{MC}}

\newcommand{\system}{\mdp}
\newcommand{\env}{\mathcal{E}}
\newcommand{\Act}{\mathsf{Act}}
\newcommand{\Envstates}{S_\env}

\newcommand{\envinit}{\iota_\env}
\newcommand{\Envact}{\mathsf{Act}_\env}
\newcommand{\envtrans}{P_\env}

\newcommand{\sensor}{\mathcal{S}}
\newcommand{\Senstates}{S_\sensor}

\newcommand{\sentrans}{P_\sensor}
\newcommand{\seninit}{\iota_\sensor}
\newcommand{{\sense}}{\mathit{sense}}
\newcommand{\obs}{z}
\newcommand{\Obs}{\mathsf{Z}}
\newcommand{\CH}{\mathsf{CH}}
\newcommand{\obsfun}{\mathsf{obs}}

\newcommand{\avact}{\mathsf{AvAct}}
\newcommand{\belief}{\mathsf{bel}}
\newcommand{\beliefs}{\mathsf{Bel}}
\newcommand{\sensbelief}{\mathsf{est_{\trsys}}}
\newcommand{\estimator}{\mathsf{est_{\mc}}}
\newcommand{\estimatorND}{\mathsf{est_{MDP}}}
\newcommand{\estimatorNDupdate}{\mathsf{est_{MDP}^{up}}}
\newcommand{\Sysstates}{S_\mathcal{J}}

\newcommand{\observationtrace}{\tau}

\newcommand{\tracerisk}[1]{R_{#1}}
\newcommand{\staterisk}{r}

\newcommand{\trace}{\tau}
\newcommand{\obsfuntrace}{\obsfun_\mathsf{tr}}
\renewcommand{\path}{\pi}

\newcommand{\finpaths}[1]{\Pi_{#1}}
\newcommand{\boundedpaths}[2]{\Pi_{#1}^{#2}}

\newcommand{\last}[1]{{#1}_\downarrow}
\newcommand{\sched}{\sigma}
\newcommand{\scheds}[1]{\textit{Sched}(#1)}
\newcommand{\dcscheds}{\Sigma_\text{DC}}

\newcommand{\prpath}[2][]{\mathsf{Pr}_{#1}(#2)}

\newcommand{\joint}{\llbracket \langle \env, \sensor \rangle \rrbracket}
\newcommand{\jointtrans}{P_{\mathcal J}}
\newcommand{\jointinit}{\init_{\mathcal J}}
\newcommand{\init}{\iota}
\newcommand{\act}{\alpha}

\newcommand{\vertices}[1]{\mathcal{V}(#1)}

\newcommand{\benchmark}[1]{\textsc{#1}}

\newcommand{\highlightPASS}[1]{\textbf{#1}}
\newcommand{\highlightFAIL}[1]{\textbf{\color{red}#1\color{black}}}
\newenvironment{problembox}{\begin{mdframed}[backgroundcolor=black!10]\textbf{The Monitoring Problem.}}{\end{mdframed}}

\makeatletter
\def\orcidID#1{\smash{\href{http://orcid.org/#1}{\protect\raisebox{-1.25pt}{\protect\includegraphics{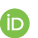}}}}}
\makeatother

\begin{document}

\mainmatter  

\title{
Runtime~Monitors~for~Markov~Decision~Processes}


%
%
\author{Sebastian Junges\orcidID{0000-0003-0978-8466} \and Hazem Torfah\orcidID{0000-0002-9628-1200} \and Sanjit A.~Seshia\orcidID{0000-0001-6190-8707}}

\authorrunning{Sebastian Junges, Hazem Torfah, and Sanjit A. Seshia}
%
\institute{University of California, Berkeley, USA\thanks{This work is partially supported by NSF grants 1545126 (VeHICaL), 1646208 and 1837132, by the DARPA contracts FA8750-18-C-0101 (AA) and FA8750-20-C-0156 (SDCPS), by Berkeley Deep Drive, and by Toyota under the iCyPhy center.}}

%
%

\maketitle
\begin{abstract}
We investigate the problem of monitoring partially observable systems with nondeterministic and probabilistic dynamics. In such systems, every state may be associated with a risk, e.g., the probability of an imminent crash. 
During runtime, we obtain partial information about the system state in form of observations. The monitor uses this information to estimate the risk of the (unobservable) current system state. 
Our results are threefold.
First, we show that extensions of state estimation approaches do not scale due the combination of nondeterminism and probabilities. 
While exploiting a geometric interpretation of the state estimates improves the practical runtime, this cannot prevent an exponential memory blowup. 
Second, we present a tractable algorithm based on model checking conditional reachability probabilities. 
Third, we provide prototypical implementations and manifest the applicability of our algorithms to a range of benchmarks. The results highlight the 
possibilities and boundaries of our novel algorithms. 

\end{abstract}

\section{Introduction}\label{sec:intro}
%
%
%
%
%
%
%
%
%
%
%

Runtime assurance is essential in deployment of safety-critical (cyber-physical) systems~\cite{DBLP:journals/fmsd/SanchezSABBCFFK19,DBLP:series/lncs/BartocciDDFMNS18,DBLP:conf/rv/HavelundR18,DBLP:conf/rv/Seshia19}. 
Monitors observe system behavior and indicate when the system is at risk to violate system specifications. 
A critical aspect in developing reliable monitors is their ability to handle noisy or missing data. In cyber-physical systems, monitors observe the system state via sensors, i.e., sensors are an interface between the system and the monitor. 
A monitor has to base its decision solely on the  obtained sensor output. 
These sensors are not perfect, and not every aspect of a system state can be measured. 


This paper considers a model-based approach to the construction of  monitors for systems with imprecise sensors. 
Consider Fig.~\ref{fig:overview}. 
We assume a model for the environment together with the controller.
Typically, such a model contains both nondeterministic and probabilistic behavior, and thus describes a Markov decision process (MDP): 
In particular, the sensor is a stochastic process~\cite{DBLP:books/daglib/0014221} that translates the environment state into an observation. 
For example, this could be a perception module on a plane that during landing  estimates the movements of an on-ground vehicle, as depicted in Fig.~\ref{fig:landing_plane}.
Due to lack of precise data, the vehicle movements itself may be most accurately described using nondeterminism.   
 
We are interested in the associated \emph{state risk} of the current system state. The state risk may encode, e.g., the probability that the plane will crash with the vehicle within a given number of steps, or the expected time until reaching the other side of the runway.
The challenge is that the monitor cannot directly observe the current system state. Instead, the monitor must infer from a trace of observations the current state risk. This cannot be done perfectly as the system state cannot be inferred precisely. Rather, we want a sound, conservative estimate of the system state. 
More concretely, for a fixed resolution of the nondeterminism, the \emph{trace risk} is the weighted sum over the probability of being in a state having observed the trace, times the risk imposed by this state.
The monitoring problem is to decide whether for any possible scheduler resolving the nondeterminism the trace risk of a given trace exceeds a threshold.

%

Monitoring of systems that contain either only probabilistic or only nondeterministic behavior is typically based on \emph{filtering}.
Intuitively, the monitor then estimates the current system states based on the model.
For purely nondeterministic systems (without probabilities) a set of states needs to be tracked, and purely probabilistic systems (without nondeterminism) require tracking a distribution over states.  This tracking is rather efficient. 
For systems that contain both probabilistic and nondeterministic behavior, filtering is more challenging.
In particular, we show that filtering on MDPs results in an exponential  memory blowup as the monitor must track sets of distributions. 
 We show that a reduction based on the geometric interpretation of these distributions is essential for practical performance, but cannot avoid the worst-case exponential blowup. 
As a tractable alternative to filtering, we rephrase the monitoring problem as the computation of conditional reachability probabilities~\cite{DBLP:conf/tacas/BaierKKM14}.
More precisely, we unroll and transform the given MDP, and then model check this MDP. 
This alternative approach yields a polynomial-time algorithm.
Indeed, our experiments show the feasibility of computing the risk by computing conditional probabilities. 
We also show benchmarks on which filtering is a competitive option.

\subsubsection*{Contribution and outline.}
This paper presents the first runtime monitoring for systems that can be adequately abstracted by a combination of \emph{probabilities and nondeterminism} and where the system state is \emph{partially observable}. 
We describe the use case, show that typical filtering approaches in general fail to deal with this setting, and show that a tractable alternative solution exists.   
In Sec.~\ref{sec:algorithms}, we investigate \emph{forward filtering}, used to estimate the possible system states in partially observable settings. We show that this approach is tractable for systems that have probabilistic \emph{or} nondeterministic uncertainty, but not for systems that have both. 
To alleviate the blowup, Sec.~\ref{sec:pruning} discusses an (often) efficacious pruning strategy and its limitations. 
In Sec.~\ref{sec:unrolling} we consider model checking as a more tractable alternative.
 This result utilises constructions from the analysis of \emph{partially observable MDPs} and model checking MDPs with \emph{conditional properties}. 
In Sec.~\ref{sec:empirical} we present baseline implementations of these algorithms, on top of the open-source model checker \storm, and evaluate their performance. 
The results show that the implementation allows for monitoring of a variety of MDPs, and reveals both strengths and weaknesses of both algorithms.
We start with a motivating example and review related work at the end of the paper. 

\subsubsection*{Motivating example.}
\label{sec:examples}

Consider a scenario where an autonomous airplane is in its final approach, i.e., lined up with a designated runway and descending for landing, see \Cref{fig:landing_plane}. On the ground, close to the runway, maintenance vehicles may cross the runway.
 The airplane tracks the movements of these vehicles and has to decide, depending on the movements of the vehicles, whether to abort the landing. 
 To simplify matters, assume that the airplane (P) is tracking the movement of one vehicle (V) that is about to cross the runway. 
 Let us further assume that P tracks V using a perception module that  can only determine the position of the vehicle with a certain accuracy~\cite{DBLP:conf/nfm/0001HTT19}, i.e., for every position of V, the perception module reports a noisy variant of the position of V. However, it is important to realize that the plane obtains a sequence of these measurements.
\begin{figure}[t!]
\centering
\hspace{0.6cm}
\subfigure[Landing plane scenario]{
\centering
\label{fig:landing_plane}
\quad
	\includegraphics[width=0.25\textwidth]{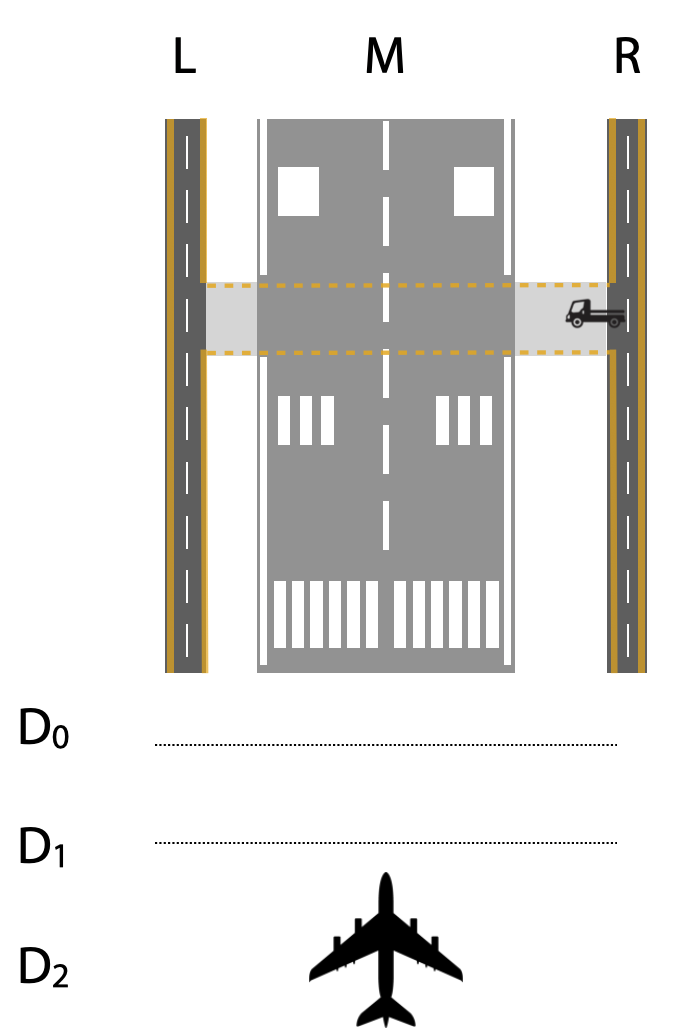}
	\quad
}
\hspace{0.4cm}
\subfigure[Sensor-based Monitors]{
\centering
\label{fig:overview}
	\scalebox{0.8}{
	\begin{tikzpicture}
		
		\node[rectangle, draw=blue, inner sep=5pt] (env) {Environment};
		\node[rectangle, draw=blue, inner sep=5pt,below=1cm of env] (con) {Controller};
		
		\node[rectangle, draw=blue, fill=blue!10!white, inner sep=5pt, below right=0.5cm of env ] (sen) {Sensors};
		
		\node[rectangle, draw=green!50!black, fill=green!10!white, inner sep=10pt,right=1.5cm of sen] (mon) {Monitor};
		\node[below=0.6cm of mon] (monres) {Risk($s$) $\geq \lambda$?};
		
		\draw[->] (env.220) |- node[left] {$s\in \Envstates$} (sen.170);
		\draw[->,dashed] (env.220) -- node[above] {} (con.140);
		
		\draw[->] (con.50) |- node[left] {} (sen.190);
		\draw[->] (con.50) -- node[above] {} (env.310);
		
		\draw[->] (sen.0) -- node[above] {$\obs \in \Obs$} (mon);
		\draw[dashed,->] (sen.0) -- +(0.5,0) |- (con.10);
		\draw[->] (mon) -- (monres);
		
		
		\node[above right = 0.7 and -0.5 of env](veh){\emph{Vehicle}};
		\path[draw,dotted,->](env) to (veh);
		
		\node[below right = 0.9 and -0 of con](plane){\emph{Plane}};
		\path[draw,dotted,->](con.340) to (plane);

		\node[above right = 0.7 and -0.5 of sen](perc){\emph{Perception}};
		\path[draw,dotted,->](sen) to (perc);
		
		\node[fit=(env)(con), draw, dotted, inner sep=5pt,label=below:{World}] {};
		\node[below = 2.5cm of sen]{};
	\end{tikzpicture}}
}
\subfigure[World Model]{
\scalebox{0.65}[0.65]{
\label{fig:worldmodel}
 \begin{tikzpicture}[rounded corners, thick, minimum width=1cm, minimum height=0.5cm]
		\node[draw,fill=red!7](MD2) at (0,0){$\langle M, D_2 \rangle$};
		\node[draw,  right= 1.3cm of MD2](RD2){$\langle R,D_2\rangle$};
		\node[draw, left=1.3cm of MD2](LD2){$\langle L, D_2 \rangle $};
		
		\node[inner sep=0pt,text width=1mm,minimum width =1mm, minimum height=1mm, left=0.7cm of RD2](RD2D2){$\bullet$};
		\node[inner sep=0pt,text width=1mm,minimum width =1mm, minimum height=1mm, left=0.7cm of MD2](MD2D2){$\bullet$};
		\node[inner sep =0pt, minimum width=1mm,minimum height=1mm, below=0.7cm  of MD2](MD2D1){$\bullet$};
		\node[inner sep =0pt, minimum width=1mm,minimum height=1mm, below=0.7cm  of RD2](RD2D1){$\bullet$};

		\node[draw, below=1.5cm of MD2, fill=red!30](MD1){$\langle M, D_1 \rangle$};
		\node[draw,  right= 1.3cm of MD1](RD1){$\langle R,D_1\rangle$};
		\node[draw, left=1.3cm of MD1](LD1){$\langle L, D_1 \rangle $};
		
		\node[inner sep=0pt,text width=1mm,minimum width =1mm, minimum height=1mm, left=0.7cm of RD1](RD1D1){$\bullet$};
		\node[inner sep =0pt, minimum width=1mm,minimum height=1mm, below=0.7cm  of MD1](MD1D0){$\bullet$};
		\node[inner sep =0pt, minimum width=1mm,minimum height=1mm, below=0.7cm  of RD1](RD1D0){$\bullet$};

		\node[draw, below=1.5cm of MD1, fill=red!60](MD0){$\langle M, D_0 \rangle$};
		\node[draw,  right= 1.3cm of MD0](RD0){$\langle R,D_0\rangle$};
		\node[draw, left=1.3cm of MD0](LD0){$\langle L, D_0 \rangle $};
		


		\path[thick,draw,->](RD2) edge [loop right] node[right]{\scriptsize $\{  w \}$} (RD2);
		\path[thick,draw,->](RD2) edge [bend left=10] node[below]{\scriptsize $\emptyset$} (RD2D2);
		\path[dashed,draw,->](RD2D2) edge [bend left=50] node[above]{$\frac{1}{20}$} (RD2.north);
		\path[dashed,draw,->](RD2D2) edge [bend right=70] node[above]{$\frac{19}{20}$} (MD2.80);
		\path[thick,draw,->](RD2) edge [bend left=10] node[left=-0.2cm]{\scriptsize$\{ p \}$} (RD2D1);
		\path[thick,draw,->](RD2.340) edge [bend left=40] node[right=-0.1cm]{\scriptsize$\{ p, w \}$} (RD1.20);
		\path[dashed,draw,->](RD2D1) edge [bend left=50] node[right=-.3cm]{$\frac{1}{2}$} (RD1.40);
		\path[dashed,draw,->](RD2D1) edge [bend right=30] node[above=0cm]{$\frac{1}{2}$} (MD1.100);

		\path[thick,draw,->](MD2) edge [bend left=10] node[below]{\scriptsize$\emptyset$} (MD2D2);
		\path[dashed,draw,->](MD2D2) edge [bend left=50] node[above]{$\frac{1}{2}$} (MD2.110);
		\path[dashed,draw,->](MD2D2) edge [bend right=70] node[above]{$\frac{1}{2}$} (LD2.north);
		\path[thick,draw,->](MD2) edge [bend left=10] node[left=-0.2cm]{\scriptsize$\{ p \}$} (MD2D1);
		\path[dashed,draw,->](MD2D1) edge [bend right=50] node[left=-.3cm]{$\frac{1}{10}$} (MD1.150);
		\path[dashed,draw,->](MD2D1) edge [bend right=30] node[above=0cm]{$\frac{9}{10}$} (LD1.80);

		\path[thick,draw,->](LD2) edge[loop left] node[left=-0.2]{\scriptsize$\emptyset$} (LD2);
		\path[thick,draw,->](LD2) edge[ bend right = 20] node[left=-0.2]{\scriptsize$\{ p \}$} (LD1);
		
		\path[thick,draw,->](RD1) edge [loop right] node[right]{\scriptsize $\{ w \}$} (RD1);
		\path[thick,draw,->](RD1) edge [bend left=10] node[below]{\scriptsize $\emptyset$} (RD1D1);
		\path[dashed,draw,->](RD1D1) edge [bend left=50] node[above]{\scriptsize$\frac{1}{2}$} (RD1.north);
		\path[dashed,draw,->](RD1D1) edge [bend right=70] node[above right=-0.2 and -0.2]{\scriptsize$\frac{1}{2}$} (MD1.40);
		\path[thick,draw,->](RD1) edge [bend left=10] node[left=-0.2cm]{\scriptsize$\{ p \}$} (RD1D0);
		\path[thick,draw,->](RD1.340) edge [bend left=40] node[right=-0.1cm]{\scriptsize$\{ p, w \}$} (RD0.east);
		\path[dashed,draw,->](RD1D0) edge [bend left=50] node[right=-.3cm]{$\frac{99}{100}$} (RD0.40);
		\path[dashed,draw,->](RD1D0) edge [bend right=30] node[above=0cm]{$\frac{1}{100}$} (MD0.80);
		\path[thick,draw,->](MD1) edge [bend right=10] node[below]{\scriptsize$\emptyset$} (LD1);
		\path[thick,draw,->](MD1) edge [bend left=10] node[left=-0.2cm]{\scriptsize$\{ p \}$} (MD1D0);
		\path[dashed,draw,->](MD1D0) edge [bend right=50] node[above left=-0.2 and -.2cm]{$\frac{1}{100}$} (MD0.120);
		\path[dashed,draw,->](MD1D0) edge [bend right=30] node[above=0cm]{$\frac{99}{100}$} (LD0.80);

		\path[thick,draw,->](LD1) edge[loop left] node[left=-0.2]{\scriptsize$\emptyset$} (LD1);
		\path[thick,draw,->](LD1) edge[ bend right = 20] node[left=-0.2]{\scriptsize$\{ p \}$} (LD0);

		\path[thick,draw,->](RD0) edge [loop below=10] node[right]{\scriptsize $^{*}$} (RD0);
		\path[thick,draw,->](MD0) edge [bend right=10] node[below]{\scriptsize $^{*}$} (LD0);
		\path[thick,draw,->](LD0) edge [loop below=10] node[right]{\scriptsize $^{*}$} (LD0);

	\end{tikzpicture}
}
}
\subfigure[Partial Sensor Model]{
\label{fig:sensormodel}
\scalebox{0.75}{
	\begin{tikzpicture}[minimum width=.4cm]
		\node[thick,draw,circle](sense)at(0,0){\small $\mathit{sense}$};
		
		\node[inner sep =0pt, minimum width=1mm,minimum height=1mm, above right = .9cm and .9cm of sense](RD2){};
		\node[above right = -0.2cm and -0.3 of RD2]{\scriptsize$\langle R,D_2\rangle/[M_o\colon\frac{1}{2}, R_o\colon\frac{1}{2} ]$};
		\path[thick,draw](sense) edge [bend right=10] node[below]{} (RD2);
		\path[thick,draw,->](RD2) edge [bend right=10] node[below]{} (sense);
		
		\node[inner sep =0pt, minimum width=1mm,minimum height=1mm, above = 1.3cm  of sense](MD2){};
		\node[above right= -0.1cm and -0.1 of MD2]{\scriptsize$\langle M,D_0\rangle/[L_o\colon\frac{1}{100} ,M_o\colon\frac{98}{100}, R_o\colon\frac{1}{100} ]$};
		\path[thick,draw](sense) edge [bend right=10] node[below]{} (MD2);
		\path[thick,draw,->](MD2) edge [bend right=10] node[below]{} (sense);
		
		
		\node[inner sep =0pt, minimum width=1mm,minimum height=1mm, right = 1.3cm  of sense](RD1){};
		\node[right = -0.2 of RD1]{\scriptsize$\langle R,D_1\rangle/[M_o\colon\frac{1}{3}, R_o\colon\frac{2}{3} ]$};
		\path[thick,draw](sense) edge [bend right=10] node[below]{} (RD1);
		\path[thick,draw,->](RD1) edge [bend right=10] node[below]{} (sense);
		
		\node[inner sep =0pt, minimum width=1mm,minimum height=1mm, below right = 0.8cm and 1.0cm  of sense](RD0){};
		\node[right = -0.2 of RD0]{\scriptsize$\langle R,D_0\rangle/[M_o\colon\frac{1}{20}, R_o\colon\frac{19}{20} ]$};
		\path[thick,draw](sense) edge [bend right=10] node[below]{} (RD0);
		\path[thick,draw,->](RD0) edge [bend right=10] node[below]{} (sense);

		\node[inner sep =0pt, minimum width=1mm,minimum height=1mm, below right= 1.3 and  0.2cm  of sense](MD1){};
		\node[below right = -0.5 and -0.1 of MD1]{\scriptsize$\langle M,D_1\rangle/[L_o\colon\frac{1}{8} ,M_o\colon\frac{3}{4}, R_o\colon\frac{1}{8} ]$};
		\path[thick,draw](sense) edge [bend right=10] node[below]{} (MD1);
		\path[thick,draw,->](MD1) edge [bend right=10] node[below]{} (sense);

		\node[inner sep =0pt, minimum width=1mm,minimum height=1mm, below left = 1.3cm and 0.3cm  of sense](MD0){};
		\node[below left = -.0 and -2.9 of MD0]{\scriptsize$\langle L,D_0\rangle/[L_o\colon\frac{19}{20}, M_o\colon\frac{1}{20} ]$};
		\path[thick,draw](sense) edge [bend right=10] node[below]{} (MD0);
		\path[thick,draw,->](MD0) edge [bend right=10] node[below]{} (sense);
		
		
	
	\draw[dotted, thick] (-.6,-.35) arc (220:150:.55);
	
	\node[below=1.3cm of sen]{};
				
	\end{tikzpicture}}
}
\caption{A probabilistic world and sensor model represented by two MDPs for the scenario of an airplane in landing approach with on-ground vehicle movements.}
\label{fig:runningExample}
\end{figure}
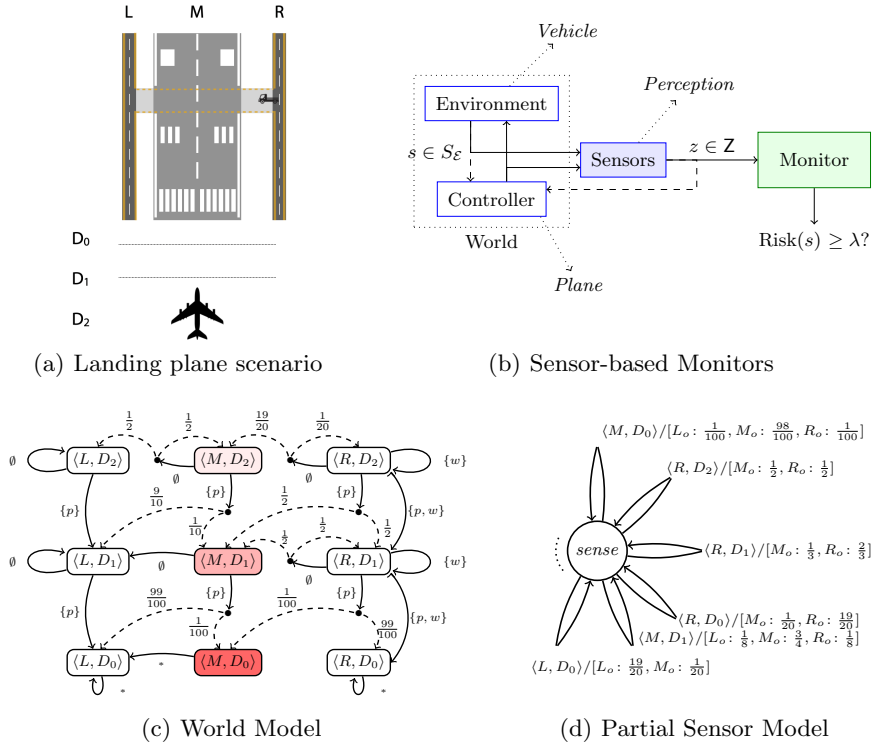

\Cref{fig:runningExample} illustrates the dynamics of the scenario. The world model describing the movements of V and P is given in \Cref{fig:worldmodel}, where $D_2,D_1$, and $D_0$ define how close P is to the runway, and $R$, $M$, and $L$ define the position of V.  Depending on what information V perceives about P, given by the atomic proposition $\{(p)\textit{rogress}\}$, and what commands it receives $\{(w)\textit{ait}\}$, it may or may not cross the runway.  The perception module receives the information about the state of the world and reports with a certain accuracy (given as a probability) the position of V. The (simple) model of the perception module is given in \Cref{fig:sensormodel}. 
For example, if P is in zone $D_2$ and V is in $R$ then there is high chance that the perception module returns that V is on the runway. The probability of incorrectly detecting V's position reduces significantly when  P is in $D_0$.

A monitor responsible for making the decision to land or to perform a go-around based on the information computed by the perception module, must take into consideration the accuracy of this returned information. For example, if the sequence of sensor readings passed to the monitor is the sequence $\trace = R_o\cdot R_o\cdot M_o$, and each state is mapped to a certain risk, then how risky is it to land after seeing $\trace$? For example, if with high probability the world is in state  $\langle M,D_0\rangle$, a very risky state, then the plane should go around. 
In the paper, we address the question of computing the risk based on this observation sequence.
We will use this example as our running example.

\section{Monitoring under Imprecise Sensors}
\label{sec:formalization}

In this section, we formalize the problem of monitoring with imprecise sensors when both the world and sensor models are given by MDPs. 
We start with a recap of MDPs, define the monitoring problem for MDPs, and finally show how the dynamics of the system under inspection can be modeled by an MDP defined by the composition of two MDPs of the sensors and world model of the system.

\subsubsection{Markov decision processes.}
For a countable set $X$, let $\Distr(X)\subset (X \rightarrow [0,1])$ define all distributions over $X$, i.e., for $d \in \Distr(X)$ it holds that $\Sigma_{x\in X} d(x) = 1$. For $\mu \in \Distr(X)$, let the \emph{support} of $\mu$ be defined by $\supp(\mu) \colonequals \{ x \mid \mu(x) > 0 \}$.
We call a distribution $\mu$ \emph{Dirac}, if $|\supp(\mu)| = 1$.

\begin{definition}[Markov decision process]
A \emph{Markov decision process} is a tuple $\mdp= \langle S,  \iota, \Act, P, 
 \Obs, \obsfun \rangle$,
where $S$ is a finite set of \emph{states}, $\iota\in \Distr(S)$ is an \emph{initial distribution}, $\Act$ is a finite set of \emph{actions}, $P\colon  S \times \Act \rightarrow \Distr(S)$ is a \emph{partial transition function},
$\Obs$ is a finite set of \emph{observations}, and $\obsfun\colon S \rightarrow \Distr({\Obs})$ is a \emph{observation function}. 
\end{definition}
\remark{The observation function can also be defined as a state-action observation function $\obsfun\colon S\times \Act \rightarrow \Distr(\Obs)$. MDPs with state-action observation function can be easily transformed into equivalent MDPs with a state observation function using auxiliary states~\cite{DBLP:journals/ai/ChatterjeeCGK16}. Throughout the paper we use state-action observations to keep (sensor) models concise.\\}

We denote $\avact(s) = \{ \act \mid P(s,\act) \neq \bot \}$. W.l.o.g., $|\avact(s)| \geq 1$.
 If all distributions in $\mdp$ are Dirac, we refer to $\mdp$ as a \emph{Kripke structure} ($\trsys$). If $|\avact(s)| = 1$ for all $s \in S$, we refer to $\mdp$ as a \emph{Markov chain} ($\mc$).
  When $\Obs = S$, we refer to $\mdp$ as \emph{fully observable} and omit $\Obs$ and $\obsfun$ from its definition. 
A \emph{finite path} in an MDP $\mdp$ is a sequence $\pi= s_0 a_0 s_1 \dots s_n\in S \times \big(\Act \times S \big)^{*}$ such that for every $0\leq i< n$ it holds that $P(s_i,a_i)(s_{i+1})>0$ and $\init(s_0) > 0$. We denote the set of finite paths of $\mdp$ by $\finpaths{\mdp}$. The \emph{length} of the path is given by the number of actions along the path. The set $\boundedpaths{\mdp}{n}$ for some $n \in \NN$ denotes the set of finite paths of length $n$.  We use  $\last{\pi}$ to denote the last state in $\pi$. We omit $\mdp$ whenever it is clear from the context.
A \emph{trace} is a sequence of observations $\trace = \obs_0 \hdots \obs_n \in \Obs^{+}$.
Every path induces a distribution over traces.
As standard, we need resolve any nondeterminism by means of a scheduler. 
\begin{definition}[Scheduler]
	 A \emph{scheduler} for MDP $\mdp$ is a function $\sched\colon \finpaths{\mdp} \rightarrow \Distr(\Act)$  with 
	 $\supp(\sched(\pi)) \subseteq \avact(\last{\path})$ for every $\pi \in \finpaths{\mdp}$.
	 \end{definition}
We use $\scheds{\mdp}$ to denote the set of schedulers. For a fixed scheduler  $\sched \in \scheds{\mdp}$, the probability $\prpath[\sched]{\path}$ of a path $\path$ (under the scheduler $\sched$) is the product of the transition probabilities in the induced Markov chain. For more details we refer the reader to~\cite{BK08}.

\subsubsection{Formal Problem Statement.}
Our goal is to determine the risk that the system is exposed to having observed a trace $\trace  \in \Obs^{+}$. Let $\staterisk \colon S\rightarrow \RRnn$ map states in $\system$ to some risk in $\RRnn$.  We call $\staterisk$ a \emph{state-risk function} for $\system$. This function maps to the risk that is associated with being in every state. 
For example, in our experiments, we flexibly define the state risk using the (expected reward extension of the) temporal logic PCTL~\cite{BK08}, to define the probability of reaching a fail state. E.g., we can define risk as the probability to crash within $H$ steps. The use of expected rewards allows for even more flexible definitions. 

Intuitively, to compute this risk of the system we need to determine the current system state having observed $\trace$ considering the probabilistic and nondeterministic context. 
Towards this, we formalize the (conditional) probabilities and risks of paths and traces.  Let $\prpath[\sched]{\path \mid \trace}$ define the probability of a path $\path$, under a scheduler $\sched$, having observed $\trace$. 
Since a scheduler may define many paths that induce the observation trace $\trace$, we are interested in the weighted risk over all paths, i.e.,  $\sum_{\path \in \boundedpaths{\system}{|\tau|}} \prpath[\sched]{\path \mid \trace} \cdot \staterisk(\last{\path})$.
The monitoring problem for MDPs then 
conservatively over-approximate the risk of a trace by assuming an adversarial scheduler, that is, by taking the supremum risk estimate over all schedulers\footnote{We later see in Lemma~\ref{lem:dcschedssuffice} that this is indeed a maximum.}.

\begin{problembox}
\label{prob:quant}
Given an MDP  $\mdp$, a state-risk $\staterisk\colon S\rightarrow \RRnn$, 
 an observation trace $\observationtrace \in \Obs^+$, 
and a threshold $\lambda \in [0,\infty)$, 
decide $\tracerisk{\staterisk}(\observationtrace) > \lambda$, 

\smallskip
\noindent where the \emph{weighted risk function} $\tracerisk{\staterisk}\colon \Obs^+ \rightarrow \RRnn$ is defined as 	\[ \tracerisk{\staterisk}(\observationtrace) \quad\colonequals\quad \sup_{\sched \in \scheds{\system}} \sum_{\path \in \boundedpaths{\system}{|\tau|}} \prpath[\sched]{\path \mid \trace} \cdot \staterisk(\last{\path}).  \]
\end{problembox}

\noindent The conditional probability $\prpath[\sched]{\path \mid \trace}$ can be characterized using Bayes' rule\footnote{for conciseness we assume throughout the paper that $\frac{0}{0}=0$}: 
\[ \prpath[\sched]{\path \mid \trace} = \frac{\prpath{\trace \mid \path} \cdot \prpath[\sched]{\path}}{\prpath[\sched]{\trace}}.\]
The probability $\prpath{\trace \mid \pi}$ of a trace $\trace$ for a fixed path $\pi$ is $\obsfuntrace(\pi)(\trace)$, where  
$$\obsfuntrace(  s ) \colonequals \obsfun(s),   \quad
 \obsfuntrace(\pi \alpha s') \colonequals \{ \trace \cdot \obs \mapsto \obsfuntrace(\pi)(\trace) \cdot \obsfun(s')(\obs) \},
$$
when $|\path| = |\trace|$, and $\obsfuntrace(\path)(\trace) = 0$ otherwise. 
The probability $\prpath[\sched]{\trace}$ of a trace $\trace$ is $ \sum_\pi \prpath[\sched]{\pi} \cdot \prpath{\trace \mid \pi}$.



 We call the special variant with $\lambda=0$ the \emph{qualitative monitoring problem}.
The problems are (almost) equivalent on Kripke structures, where considering a single path to an adequate state suffices. Details are given in the appendix. 
\begin{lemma}
\label{lem:qualonkripke}
	For Kripke structures the monitoring and qualitative monitoring problems are logspace interreducible. 
\end{lemma}

In the next sections we present two types of algorithms for the monitoring problem. 
The first algorithm is based on the widespread (forward) filtering approach \cite{forwardfiltering}. The second is new algorithm based on model checking conditional probabilities. While filtering approaches are efficacious in a purely nondeterministic or a purely probabilistic setting, it does not scale on models such as MDPs that are both probabilistic and nondeterministic. In those models, model checking provides a tractable alternative.  However, we first connect the problem statement more formally to our motivating example.

\subsubsection{An MDP defining the system dynamics.}
Before we solve the monitoring problem for MDPs, we show how the weighted risk for a system given by a world and sensor model can be formalized as a monitoring problem for MDPs. To this end, we define the dynamics of the world and sensors that we use as basis for our monitor as the following joint MDP.

For a fully observable world MDP $\env=\langle \Envstates, \envinit, \Envact, \envtrans 
\rangle $ and a sensor MDP $\sensor=\langle \Senstates, \seninit, \Envstates, \sentrans, \Obs, \obsfun \rangle$, where $\obsfun$ is state-action based, 
the \emph{inspected system} is defined by an MDP $ \joint = \langle \Sysstates, \jointinit, \Envact, \jointtrans,
 \Obs, \obsfun_{\mathcal J} \rangle$ being the synchronous composition of $\env$ and $\sensor$:
\begin{compactitem}
\item $\Sysstates \colonequals \Envstates \times \Senstates$,
\item $\jointinit$ is defined as $\jointinit(\langle u,s \rangle) \colonequals \envinit(u)\cdot \seninit(s)$ for each $u\in \Envstates$ and $s\in \Senstates$,
\item  $\jointtrans \colon \Sysstates \times \Envact \rightarrow \Distr(\Sysstates)$ such that 
for all $\langle u,s \rangle \in \Sysstates$ and $\act \in \Envact$;
\[ \jointtrans(\langle u,s\rangle, \act) = 	 d_{u,s} \in \Distr(\Sysstates),
\]
where for all $u'\in \Envstates$ and $s'\in \Senstates$:
$d_{u,s}(\langle u',s'\rangle)=  \envtrans(u,\act)(u') \cdot \sentrans(s,u)(s')$,
\item $\obsfun_{\mathcal J}\colon \Sysstates \rightarrow \Distr(\Obs)$ with $\obsfun_\mathcal{J}\colon \langle u,s\rangle \mapsto \obsfun(s,u)$.
\end{compactitem}

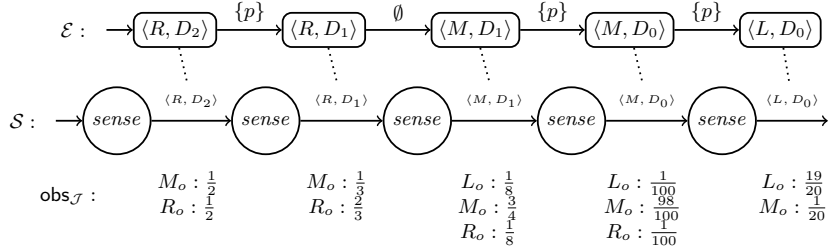
\begin{figure}[t]

\scalebox{0.85}[0.85]{
\begin{tikzpicture}[rounded corners, thick, minimum width=1cm, minimum height=0.5cm]
	\node[draw,initial, initial text= $\env:$] (u0) {$\langle R,D_2\rangle$};
	\node[draw,right=of u0] (u1) {$\langle R,D_1\rangle$};
	\node[draw,right=of u1] (u2) {$\langle M, D_1 \rangle$};
	\node[draw,right=of u2] (u3) {$\langle M,D_0 \rangle$};
	\node[draw,right=of u3] (u4) {$\langle L, D_0\rangle$};
	\node[initial, initial text=$\sensor:$,state,below=0.6cm of u0, xshift=-0.9cm] (s0) {\small $\mathit{sense}$};
	\node[state,below=0.6cm of u1, xshift=-0.9cm] (s1) {\small $\mathit{sense}$};
	\node[state,below=0.6cm of u2, xshift=-0.9cm] (s2) {\small $\mathit{sense}$};
	\node[state,below=0.6cm of u3, xshift=-0.9cm] (s3) {\small $\mathit{sense}$};
	\node[state,below=0.6cm of u4, xshift=-0.9cm] (s4) {\small $\mathit{sense}$};
	\node[right=2cm of s4, xshift=-0.9cm] (s5) {};
	
	\node[below=1.8cm of u0, xshift= 0.2cm, align= center] (z0) { $M_o:\frac{1}{2}$\\$R_o:\frac{1}{2}$};
	\node[below=1.8cm of u1, xshift= 0.2cm, align= center] (z1) {$M_o:\frac{1}{3}$\\$R_o:\frac{2}{3}$};
	\node[below=1.8cm of u2, xshift= 0.2cm, align= center] (z2) {$L_o: \frac{1}{8}$\\$M_o:\frac{3}{4}$\\$R_o:\frac{1}{8}$};
	\node[below=1.8cm of u3, xshift= 0.2cm, align= center] (z3) {$L_o: \frac{1}{100}$\\$M_o:\frac{98}{100}$\\$R_o:\frac{1}{100}$};
	\node[below=1.8cm of u4, xshift= 0.2cm, align= center] (z4) {$L_o: \frac{19}{20}$\\$M_o:\frac{1}{20}$};

	\draw[->] (u0) -- node[above] {$\{p\}$} (u1);
	\draw[->] (u1) -- node[above] {$\emptyset$} (u2);
	\draw[->] (u2) -- node[above] {$\{p\}$} (u3);
	\draw[->] (u3) -- node[above] {$\{p\}$} (u4);
	\draw[->] (s0) -- node (a0)[above=.05] {\tiny $\langle R,D_2\rangle$} (s1);
	\draw[dotted] (u0) -- (a0);
	\draw[->] (s1) -- node (a1)[above=.05] {\tiny $\langle R,D_1\rangle$} (s2);
	\draw[dotted] (u1) -- (a1);
	\draw[->] (s2) -- node (a2) [above=.05] {\tiny $\langle M, D_1 \rangle$} (s3);
	\draw[dotted] (u2) -- (a2);
	\draw[->] (s3) -- node (a3) [above=.05] {\tiny $\langle M,D_0 \rangle$} (s4);
	\draw[dotted] (u3) -- (a3);
	\draw[->] (s4) -- node (a4) [above=.05] {\tiny $\langle L, D_0\rangle$} (s5);
	\draw[dotted] (u4) -- (a4);

	\node[anchor=east, left=-0.3 of z0 ,xshift= -1cm] {$\obsfun_{\mathcal J}:$};
\end{tikzpicture}
}
\caption{A run with its observations of the inspected system $\joint$ where $\env$ and $\sensor$ are the models given in \Cref{fig:runningExample}. }
\label{fig:paths}
\end{figure}

In \Cref{fig:paths} we illustrate a run of $\joint$  for the world and sensor MDPs presented in \Cref{fig:runningExample}. We particularly show the observations of the joint MDP given by the distributions over the observations  for each transition in the run (we omitted the probabilistic transitions for simplicity).  
The observations  of the MDP $\system$ present the output of the sensor upon a path through $\mdp$. These observations in turn are the inputs to a monitor on top of the system. The role of the monitor is then to compute the risk of being in a critical state based on the previously received observations. 

%

\section{Forward Filtering for State Estimation}
\label{sec:algorithms}

We start by showing why standard forward filtering does not scale well on MDPs. 
We briefly show how filtering can be used to solve the monitoring problem for  purely nondeterministic systems  (Kripke structures) or purely probabilistic systems  (Markov Chains). Then, we show why for MDPs, the forward filtering needs to manage, although finite but an exponential set of distributions. In \Cref{sec:pruning} we present a new improved variant of forward filtering for MDPs based on filtering with vertices of the convex hull. In \Cref{sec:unrolling} we present a new polynomial-time model checking-based algorithm for solving the problem. 



\subsection{State estimators for Kripke structures.}
\label{sec:operationalmodel}

For Kripke structures, we maintain a set of possible states that agree with the observed trace. This set of states is  inductively characterized by the function $\sensbelief\colon \Obs^+ \rightarrow 2^{S}$ which we define formally below. 
For an observation trace $\trace$, $\sensbelief(\trace)$ defines the set of states that can be reached with positive probability. This set can be computed by a forward state traversal \cite{forwardrechability}. 
To illustrate how  $\sensbelief(\trace)$ is computed for $\trace$, 
consider the underlying Kripke structure of the inspected system $\joint$ for our running example in \Cref{fig:runningExample} (to make this a Kripke structure, we remove the probabilities). Consider further the observation trace $\trace = R_o \cdot M_o \cdot L_o$. 
 Since $\joint$ has only one initial state $\langle \langle R, D_2\rangle, \mathit{sense}\rangle$ and $R_o$ is observable with a positive probability in this state, $\sensbelief(R_o)= \{\langle \langle R, D_2\rangle, \mathit{sense}\rangle\}$. 
As $M_o$ is observed next, $\sensbelief(R_o \cdot M_o)$ computes the states reached from  $\langle \langle R, D_2\rangle, \mathit{sense}\rangle$ and where $M_o$ can be observed with a positive probability, i.e., $\sensbelief(R_o M_o)= \{ \langle \langle R, D_1\rangle, \mathit{sense}\rangle, \langle \langle R, M_1\rangle, \mathit{sense}\rangle\}$. 
Finally the current state having observed $R_o \cdot M_o \cdot L_o$ may be one of the the states $\sensbelief(\trace)= \{ \langle \langle M, D_1\rangle, \mathit{sense}\rangle$, $\langle \langle L, D_1\rangle, \mathit{sense}\rangle$, $\langle \langle L, D_0\rangle, \mathit{sense}\rangle$, $\langle \langle M, D_0\rangle, \mathit{sense}\rangle\}$, which especially shows that we might be in the high-risk world state $\langle M, D_0\rangle$. 

\begin{definition}[$\trsys$ state estimator]\label{def:ksstateest}
For $\trsys= \langle S,  \iota, \Act, P,
 \Obs, \obsfun \rangle$, the state estimation function $\sensbelief \colon \Obs^+ \rightarrow 2^{S}$ is defined as   \vspace{-.1cm}
\begin{align*}
 	&\sensbelief(z) \colonequals \{s \in S \mid \init(s)>0 \wedge \obsfun(s)(z) > 0 \}\\
	&\sensbelief(\trace\cdot \obs)  \colonequals \Big\{ s' \in S \mid \exists s \in \sensbelief(\trace), \exists \act \in \Act, P(s,\act)(s') > 0 \wedge \obsfun(s')(z)>0 \Big\}.
\end{align*}
\end{definition}

For a Kripke structure $\trsys$ and a given trace $\trace$, the monitoring problem can be solved by computing $\sensbelief(\trace)$, using \cite{forwardrechability} and \Cref{lem:qualonkripke}.

\begin{lemma}
\label{lem:kripkestructest}
 For a Kripke stucture  $\trsys=\langle S,  \iota, \Act, P,
 \Obs, \obsfun \rangle$, a trace $\observationtrace\in \Obs^+$, and a state-risk function $r\colon S \rightarrow \RRnn $, it holds that $\tracerisk{\staterisk}(\observationtrace) = \max \limits_{s\in \sensbelief(\trace)} \staterisk(s)$. Computing $\tracerisk{\staterisk}(\observationtrace)$  requires time $\mathcal{O}(|\tau| \cdot |P|)$  and space  $\mathcal{O}(|S|)$. 
\end{lemma}
A more detailed proof can be found in the appendix. 

	The time and space requirements follow directly from the inductive definition of $\sensbelief$ which resembles solving a forward state traversal problem in automata~\cite{forwardrechability}. In particular, the algorithm allows updating the result after extending $\tau$ in $\mathcal{O}(|P|)$. 

\subsection{State estimators for Markov chains.}
For Markov chains, in addition to tracking the potential reachable system states, we also need to take the transition probabilities into account.
When a system is (observation-)deterministic, we can adapt the notion of beliefs, similar to RVSE~\cite{DBLP:conf/rv/StollerBSGHSZ11}, and similar to the construction of belief MDPs for \emph{partially observable MDPs}, cf.~\cite{DBLP:books/sp/12/Spaan12}:

\begin{definition}[Belief]
	For an MDP $\system$ with a set of states $S$, a belief~$\belief$ is a distribution in $ \Distr(S)$. 
\end{definition}
In the remainder of the paper, we will denote the function $S\rightarrow \{0\}$ by $\mathbf{0}$ and the set $\Distr(S)\cup \{\mathbf{0}\}$ by $\beliefs$. A state estimator based on $\beliefs$ is then defined as follows~\cite{DBLP:conf/vmcai/SistlaS08,RTVStochasticFaultySystems,DBLP:conf/rv/StollerBSGHSZ11}\footnote{For the deterministic case, we omit the unique action for brevity}:

\begin{definition}[MC state estimator]
\label{def:mcstateest}
For $\mc =\langle S,  \iota, \Act, P,
 \Obs, \obsfun \rangle $, a trace $\trace \in \Obs^+$ the state estimation function $\estimator \colon \Obs^+ \rightarrow \beliefs$ is defined as 
 \begin{align*}
 	\estimator(\obs) &\colonequals \begin{cases}
 									\big \{ s \mapsto \frac{\init(s)\cdot \obsfun(s)(z)}{\sum \limits_{\hat{s} \in S} \init(\hat{s})\cdot \obsfun(\hat{s})(z)} \big \}& \exists s \in S.~\init(s)\cdot \obsfun(z)>0,\\
 									\mathbf{0} & \text{otherwise}.
 								   \end{cases}
 	\\
 	\estimator(\trace\cdot \obs) &\colonequals \left\{ s' \mapsto \frac{ \sum \limits_{s\in S} \estimator(\trace)(s) \cdot P (s,s')\cdot \obsfun(s')(\obs)}{\sum \limits_{s\in S }\estimator(\trace)(s) \cdot \big(\sum \limits_{\hat{s}\in S} P(s,\hat{s}) \cdot \obsfun(\hat{s})(\obs) \big)  } \right\}
 \end{align*}

\end{definition}

To illustrate how $\estimator$ is computed, consider again our system in \Cref{fig:runningExample} and assume that the MDP has only the actions labeled with $\{p\}$ (reducing it to the Markov chain induced by the a scheduler that only performs the $\{p\}$ actions). Again we consider the observation trace $\trace= R_o \cdot M_o \cdot L_o$ and compute $\estimator{\trace}$.  
For the first observation $R_o$, and since there is only one initial state, it follows that $\estimator(R_o) =   \{\langle R,D_2\rangle \mapsto 1\}$\footnote{We omit the (single) sensor state for conciseness.}. 
From $\langle R,D_2\rangle$ and having observed $M_o$ we can reach the states $\langle R,D_1\rangle$ and $\langle M,D_1\rangle$ with probabilities $\estimator(R_o\cdot M_o) =  \{\langle R,D_1\rangle \mapsto \frac{\frac{1}{2}\cdot \frac{1}{3}}{\frac{1}{2}\cdot \frac{1}{3} + \frac{1}{2}\cdot \frac{3}{4}} = \frac{4}{13} , \langle M,D_1\rangle \mapsto \frac{\frac{1}{2}\cdot \frac{3}{4}}{\frac{1}{2}\cdot \frac{1}{3} + \frac{1}{2}\cdot \frac{3}{4}} = \frac{9}{13}\}$.
 Finally, from the later two states, when observing $L_o$, the states $\langle M,D_0\rangle$ and $\langle L,D_0\rangle$ can be reached with probabilities $ \estimator(R_o\cdot M_o \cdot L_o) =  \{ \langle M,D_0\rangle \mapsto 0.0001,  \langle L,D_0\rangle\mapsto 0.999\}$. 
 Notice that although the state $\langle R,D_0\rangle$ can be reached from $\langle R,D_1\rangle$, the probability of being in this state is 0 since observing $L_o$  in this state is $\obsfun(\langle R,D_0\rangle)(L_o)=0$. 


\begin{lemma}
\label{lem:mcfilteringcorrect}
 For a Markov chain  $\mc =\langle S,  \iota, \Act, P,
 \Obs, \obsfun \rangle$, a trace $\observationtrace\in \Obs^+$, and a state-risk function $r\colon S \rightarrow \RRnn $, it holds that
$\tracerisk{\staterisk}(\observationtrace) = \sum_{s \in S} \estimator(\trace)(s) \cdot \staterisk(s). $
Computing $\tracerisk{\staterisk}(\observationtrace)$  can be done in time $\mathcal{O}(|\trace|\cdot|S|\cdot|P|)$ , and  using $|S|$ many rational numbers. The size of the rationals\footnote{To avoid growth, one may use fixed-precision numbers that over-approximate the probability of being in any state---inducing a growing (but conservative) error.} may grow linearly in $\trace$. 
\end{lemma}

\smallskip
\noindent\emph{Proof sketch.}
 Since the system is deterministic, there is a unique scheduler $\sigma$, thus $\tracerisk{\staterisk}(\observationtrace) = \sum_{\path \in \boundedpaths{\mc}{|\tau|}} \prpath[\sched]{\path \mid \trace} \cdot \staterisk(\last{\path})$ by definition.  
 We can show by induction over the length of $\trace$ that $\prpath[\sched]{\path \mid \trace} = \estimator(\trace)(\last{\path}) $ and conclude that $\tracerisk{\staterisk}(\observationtrace) =\sum_{\path \in \boundedpaths{\system}{|\tau|}}  \estimator(\trace)(\last{\path})  \cdot \staterisk(\last{\path}) = \sum_{s \in S} \estimator(\trace)(s) \cdot \staterisk(s)$ because $\estimator(\trace)(s) = 0 $ for all $s \in S$ for which there is no path $\path \in \boundedpaths{\system}{|\tau|}$ with $\last{\pi} = s$. 
	The complexity follows from the inductive definition of $\estimator$ that requires in each inductive step  to iterate over all transitions of the system and maintain a belief over the states of the system.  \qed

\subsection{State estimators for Markov decision processes.}
In an MDP, we have to account for every possible resolution of nondeterminism, which means that a belief can evolve into a set of beliefs:

\begin{definition}[MDP state estimator]
\label{def:ndprobupdate}
For an MDP $\mdp =\langle S,  \iota, \Act, P,
 \Obs, \obsfun \rangle$, a trace $\observationtrace\in \Obs^+$, and a state-risk function $r\colon S \rightarrow \RRnn $, the state estimation function  $\estimatorND \colon \Obs^+ \rightarrow 2^\beliefs$ is defined as 
\begin{alignat*}{3}
	&\estimatorND(\obs) &=& \{ \estimator(z) \}, \\
	&\estimatorND(\trace\cdot \obs) &=&  
	\Big\{ \belief' \in \beliefs \Bigm| \exists  \belief \in \estimatorND(\trace).~ \belief' \in \estimatorNDupdate(\belief, \obs) \Big\},
	\end{alignat*}
	and where $\belief' \in \estimatorNDupdate(\belief, \obs)$ if there exists $\varsigma_\belief \colon S \rightarrow \Distr(\Act)$ such that:
	\[\forall s'. \belief'(s') = \frac{\sum\limits_{s\in S} \belief(s) \cdot \sum \limits_{\act \in \Act} \varsigma_\belief(s)(\act) \cdot   P(s,\act,s') \cdot \obsfun(s')(\obs)}{\sum \limits_{s\in S} \belief(s) \cdot \sum \limits_{\act \in \Act} \varsigma_\belief(s)(\act) \cdot \ \sum \limits_{\hat{s}\in S} P(s,\act,\hat{s}) \cdot \obsfun(\hat{s})(\obs)    }.\]
\end{definition}

\noindent
The definition conservatively extends both Def.~\ref{def:ksstateest} and Def.~\ref{def:mcstateest}.
Furthermore, we remark that we do not restrict how the nondeterminism is resolved: any distribution over actions can be chosen, and the distributions may be different for different traces.

Consider our system in \Cref{fig:runningExample}. For the trace $\trace = R_o\cdot M_o\cdot L_o$, $\estimatorND(\trace)$ is computed as follows. First, when observing $R_o$, the state estimator computes the initial belief set $\estimatorND(R_o)= \{ \{\langle R,D_2\rangle \mapsto 1\}\}$. From this set of beliefs, when observing $M_o$, a set $\estimatorND(R_o\cdot M_o)$ can be computed since all transitions $\emptyset, \{p\},\{w\},\{p,w\}$ (as well as their convex combinations) are possible from $\langle R,D_2\rangle$. 
One of these beliefs is for example $\{\langle R,D_1\rangle \mapsto   \frac{4}{13} , \langle M,D_1\rangle \mapsto  \frac{9}{13}\}$ when a scheduler takes the transition $\{p\}$ (as was computed in our example for the Markov chain case). 
Having additionally observed $L_o$ a new set $ \estimatorND(R_oM_oL_o)$ of beliefs can be computed based on the beliefs in $\estimatorND(R_oM_o)$. 
For example from the belief $\{\langle R,D_1\rangle \mapsto   \frac{4}{13} , \langle M,D_1\rangle \mapsto  \frac{9}{13}\}$, two of the new beliefs are  $\{\langle L,D_0\rangle \mapsto  0.999 , \langle M,D_0\rangle \mapsto 0.0001\}$ and $\{\langle M,D_1\rangle \mapsto 0.0287, \langle M, D_0 \rangle \mapsto 0.0001, \langle L, D_0 \rangle \mapsto 0.9712\}$. The first belief is reached by a scheduler that takes a transition $\{p\}$ at both $\langle R, D_1 \rangle$ and $\langle M, D_1 \rangle$. Notice that the belief does not give a positive probability to the state $\langle R, D_0\rangle$ because  $L_o$ cannot be observed in this state. The second belief is reached by considering a scheduler that takes transition $\{p\}$ at $\langle M, D_1\rangle$ and transition $\emptyset$ at $\langle R, D_1\rangle$. 

\begin{theorem}
\label{thm:mdpthm} For an MDP  $\mdp =\langle S,  \iota, \Act, P,
 \Obs, \obsfun \rangle$, a trace $\observationtrace\in \Obs^+$, and a state-risk function $r\colon S \rightarrow \RRnn $, it holds that $\tracerisk{\staterisk}(\observationtrace) = \sup_{\belief \in \estimatorND(\observationtrace)} \sum_{s \in S} \belief  (s) \cdot \staterisk(s) $. 
\end{theorem}

\smallskip
\noindent\emph{Proof sketch.}
For a given trace $\trace$, each (history-dependent, randomizing) scheduler  induces a belief over the states of the Markov chain induced by the scheduler. Also, each belief in $\estimatorND(\trace)$ corresponds to a fixed scheduler, namely that one used to compute the belief recursively (i.e., an arbitrary randomizing memoryless scheduler for every time step). 
Once a scheduler $\sched$ and its corresponding belief $\belief$ is fixed, or vice versa, we can show using induction over the length of $\trace$ that $\sum_{\path \in \boundedpaths{\system}{|\tau|}} \prpath[\sched]{\path \mid \trace} \cdot \staterisk(\last{\path})  = \sum_{s \in S} \belief(s)\cdot \staterisk(s)$. 
 \qed

\section{Convex Hull-based Forward Filtering}
\label{sec:pruning}
In this section, 
we show that we can use a finite representation for  $\estimatorND(\observationtrace)$, but that this representation is exponentially large for some MDPs.

\subsection{Properties of $\estimatorND(\trace)$.} 
First, observe that $\mathbf{0}$ never maximises the risk. Furthermore, $\estimatorNDupdate(\mathbf{0}, \obs) =\{\mathbf{0}\}$. We can thus w.l.o.g. assume that $\mathbf{0} \not\in \estimatorND(\trace)$.
Second, observe that $\estimatorND(\trace) \neq \emptyset$ if $\prpath[\sched]{\trace} > 0$.

We can interpret a belief $\belief \in \beliefs$ as point in (a bounded subset of) $\RR^{(|S|-1)}$. We are in particular interested in convex sets of beliefs. 
A set $B \subseteq \beliefs$ is convex if the convex hull $\CH(B)$ of $B$, i.e. all convex combination of beliefs in $B$
\footnote{That is, $\CH(B) =  \{ \sum_{\belief \in B}\sum_{s \in S}\; w_s \cdot \belief(s) \mid \text{for all } w_s \in \RRnn \text{ with }\sum w_s = 1 \}$.}, 
coincides with $B$, $\CH(B) = B$. 
For a set $B \subseteq \beliefs$, a belief $\belief \in B$ is an interior belief if it can be expressed as convex combination of the beliefs in $B \setminus \{ \belief \}$. 
All other beliefs are (extremal) points or \emph{vertices}. 
Let the set $\vertices{B} \subseteq B$ denote the set of \emph{vertices of the convex hull} of  $B$. 

\begin{figure}[t]
\subfigure[]{
\begin{tikzpicture}[every node/.style={font=\scriptsize},simpstate/.style={circle,draw,fill=white,inner sep=1pt,minimum width=14pt},act/.style={circle,fill=black,inner sep=1pt}]

	\node[simpstate,fill=blue!30] (s1) {$s_1$};
	
	\node[simpstate,left=0.6cm of s1,fill=blue!30] (s0) {$s_0$};	
	\node[simpstate,below=of s0,fill=blue!30] (s3) {$s_3$};
	\node[simpstate,below=of s1,fill=orange!30] (s2) {$s_2$};
	\node[simpstate,below = 0.5 cm of s2,fill=orange!30] (s4) {$s_4$};
	\node[left=0.2cm of s0] (start) {};
	\draw[->,black] (start) -- (s0);
	\node[act,right=0.3cm of s0, fill=black] (a01) {};
	\node[act,below=0.3cm of s0, fill=black] (a03) {};
	\draw[black] (s0) -- (a01);
	\draw[black] (s0) -- (a03);
	
	\draw[->,black] (a01) edge node[auto] {$1$} (s1);
	\draw[->,black] (a03) -- node[right, pos=0.4] {$\nicefrac{1}{4}$} (s3);
	\draw[->,black] (a03) -- node[below, pos=0.6] {$\nicefrac{3}{4}$} (s1);

	\node[act,above=0.2cm of s1,fill=black] (a11) {};
	\node[act,below=0.3cm of s1,fill=black] (a12) {};
	\node[act,right=0.3cm of s1,fill=black] (a13) {};
	\node[act,right=0.3cm of s3,fill=black] (a31) {};
	\node[act,left=0.3cm of s3,fill=black] (a32) {};
	
	\draw[black] (s1) -- (a11);
	\draw[black] (s1) -- (a12);
	\draw[black] (s1) -- (a13);
	\draw[black] (s3) -- (a31);
	\draw[black] (s3) -- (a32);
	
	\draw[->,black] (a13) edge[bend right] node[pos=0.35,above] {$1$} (s1);
	\draw[->,black] (a11) edge[bend right] node[pos=0.35,above] {$1$} (s0);
	\draw[->,black] (a12) edge[bend right] node[pos=0.35,right] {$\nicefrac{1}{2}$} (s1);
	\draw[->,black] (a12) edge node[pos=0.35,right] {$\nicefrac{1}{2}$} (s3);

	\draw[->,black] (a31) edge[bend right] node[pos=0.7,below] {$\nicefrac{1}{2}$} (s4);
	\draw[->,black] (a31) edge node[pos=0.5,below] {$\nicefrac{1}{2}$} (s2);
	\draw[->,black] (a32) edge[bend right] node[pos=0.35,below] {$\nicefrac{1}{2}$} (s3);
	\draw[->,black] (a32) edge[bend left=18] node[pos=0.3,left] {$\nicefrac{1}{2}$} (s0);

\end{tikzpicture}
\label{fig:expbeliefs:mdp}
	}
	\subfigure[$B$ over $s_1, s_3$]{
	\begin{tikzpicture}[every node/.style={font=\scriptsize}]
\node at (-0.58,0) {};
\node at (1,0) {};

	\draw[thick,dotted] (0,0) -- (0,2);
	\node[circle,inner sep=2pt, fill=blue!60] at (0,2) {};
	\node[circle,inner sep=2pt, fill=blue!60] at (0,1.5) {};
	
	\node[anchor=south] at (0,2) {$1,0$};
	\node[anchor=north] at (0.5,1.5) {$\nicefrac{3}{4},\nicefrac{1}{4}$};
\end{tikzpicture}
	\label{fig:expbeliefs:1}
	}
	\subfigure[$B$ over $s_0, s_1, s_3$]{
	\begin{tikzpicture}[every node/.style={font=\scriptsize}]
	\draw[thick,dotted]  (2,0)  -- (0,0) -- (0,0.25) -- (0.75,1) -- (1,1) -- (2,0);
	
	\draw[black!70,dotted] (0,0) -- (0,2) -- (2,0) -- (0,0);
	\node[circle,inner sep=2pt, fill=blue!60] at (2,0) {}; 
	\node[circle,inner sep=2pt, fill=blue!60] at (0,0) {}; 
	\node[circle,inner sep=2pt, fill=blue!60] at (1,1) {}; 

	\node[circle,inner sep=2pt, fill=blue!60] at (1.5,0.25) {}; 
	\node[circle,inner sep=2pt, fill=blue!60] at (0,0.25) {}; 
	\node[circle,inner sep=2pt, fill=blue!60] at (0.75,1) {}; 
	
	

	
	

	\node[anchor=north] at (0.0,0.0) {$1,0,0$};
	\node[anchor=north] at (2.0,0.0) {$0,1,0$};
	\node[anchor=west] at (1.0,1.0) {$0,\nicefrac{1}{2},\nicefrac{1}{2}$};
	\node[anchor=east,yshift=2mm] at (1.5,0.25) {$\nicefrac{1}{8}, \nicefrac{3}{4},\nicefrac{1}{8}$};
	\node[anchor=east] at (0.0,0.25) {$\nicefrac{7}{8},0,\nicefrac{1}{8}$};
	\node[anchor=east] at (0.75,1.0) {$\nicefrac{1}{8},\nicefrac{3}{8},\nicefrac{1}{2}$};
	
\end{tikzpicture}
		\label{fig:expbeliefs:2}
	}
	\subfigure[$B$ over $s_2, s_3$]{
	\label{fig:expbeliefs:3}

	\begin{tikzpicture}[every node/.style={font=\scriptsize}]
\node at (-0.58,0) {};
\node at (1,0) {};

	\draw[thick,dotted] (0,0) -- (0,2);
	\node[circle,inner sep=2pt, fill=orange!60] at (0,1) {};
	
	\node[anchor=north] at (0.5,1.5) {$\nicefrac{1}{2},\nicefrac{1}{2}$};
\end{tikzpicture}
	}	
	\caption{Beliefs in $\RR^n$ on $\mdp$ for $\tau = \textcolor{blue}{\obs_0\obs_0}$, $\textcolor{blue}{\obs_0\obs_0\obs_0}$ and $\textcolor{blue}{\obs_0\obs_0}\textcolor{orange}{\obs_1}$, respectively.}
	\label{fig:expbeliefs}
\end{figure}
\begin{example}
	Consider Fig.~\ref{fig:expbeliefs:mdp}. All observation are Dirac, and only states $s_2$ and $s_4$ have observation $\textcolor{orange}{\obs_1}$.  
	The beliefs having observed $\textcolor{blue}{\obs_0\obs_0}$ are distributions over $s_1,s_3$, 
	and can thus be depicted in a one-dimensional simplex. In particular, we have $\vertices{\estimatorND(\textcolor{blue}{\obs_0\obs_0})} = \left\{ \{ s_1 \mapsto 1 \},\{ s_1 \mapsto \nicefrac{3}{4}, s_3 \mapsto \nicefrac{1}{4} \} \right\}$, as depicted in Fig.~\ref{fig:expbeliefs:1}. 
	The six beliefs having observed $\textcolor{blue}{\obs_0\obs_0\obs_0}$ are distributions over $s_0,s_1,s_3$, depicted in Fig.~\ref{fig:expbeliefs:2}. Five out of six beliefs are vertices.   
	The belief having observed $\textcolor{blue}{\obs_0\obs_0}\textcolor{orange}{\obs_1}$ is in Fig.~\ref{fig:expbeliefs:3}.
\end{example}

\begin{remark}
Observe that we illustrate the beliefs over only the states $\sensbelief(\tau)$. We therefore call $|\sensbelief(\observationtrace)|$ the dimension of $\estimatorND(\observationtrace)$. 
\end{remark}
From the fundamental theorem of linear programming~\cite[Ch.~7]{DBLP:books/daglib/0090562} it immediately follows that the trace risk $\tracerisk{\trace}$ is obtained at a vertex of the beliefs of $\estimatorND{\trace}$. We obtain the following refinement over Theorem~\ref{thm:mdpthm}:
\begin{theorem}
For every $\observationtrace$ and $\staterisk$: 
$\tracerisk{\staterisk}(\observationtrace) = \max \limits_{\belief \in \textcolor{black}{\vertices{\textcolor{black}{\estimatorND(\observationtrace)}}}} \sum_{s \in S} \belief  (s) \cdot \staterisk(s). $
\label{theo:chcomp}
\end{theorem}
Lemma~\ref{lem:finitevertices} below clarifies that this maximum indeed exists.

We make some observations that allow us to compute the vertices more efficiently: 
Let $\estimatorNDupdate(B,\obs)$ denote $\bigcup_{\belief \in B} \estimatorNDupdate(\belief,\obs)$. 
From the properties of convex sets~\cite[Ch.~2]{DBLP:books/cu/BV2014}, we make the following observations:
If $B$ is convex $\estimatorNDupdate(B,\obs)$ is convex, as all operations in computing a new belief are convex-set preserving\footnote{The scaling is called a \emph{projection}.}. 
Furthermore, if $B$ has a finite set of vertices, then $\estimatorNDupdate(B,\obs)$ has a finite set of vertices.
The following lemma which is based on the observations above clarifies how to compute the vertices:
\begin{lemma}
For a convex set of beliefs $B$  with a finite set of vertices and an observation $\obs$:
\label{lem:convexsuffices}
	\[ \vertices{\estimatorNDupdate(B,\obs)} = \vertices{\estimatorNDupdate(\vertices{B},\obs)}. \]
\end{lemma}
By induction and using the facts above we obtain:
\begin{lemma}
\label{lem:finitevertices}
Any $\vertices{\estimatorND(\trace)}$ is finite.	
\end{lemma}
A monitor thus only needs to track the vertices. Furthermore, $\estimatorNDupdate(B,\obs)$ can be adapted to compute only vertices by limiting $\varsigma_\belief$ to $S \rightarrow \Act$.
\subsection{Exponential lower bounds on the relevant vertices.}
We show that a monitor in general cannot avoid an exponential blow-up in the beliefs it tracks.
First observe that updating $\belief$ yields up to  $\prod_s |\Act(s)|$ new beliefs (vertex or not), a prohibitively large number. The number of vertices is also exponential:
\begin{lemma}
\label{lem:exponentialbelief}
There exists a family of MDPs $\mdp_{n}$ with $2n+1$ states such that 
	 ${|\vertices{\estimatorND(\observationtrace)}|=2^n}$	 for every $\trace$ with $|\trace| > 2$.
\end{lemma}
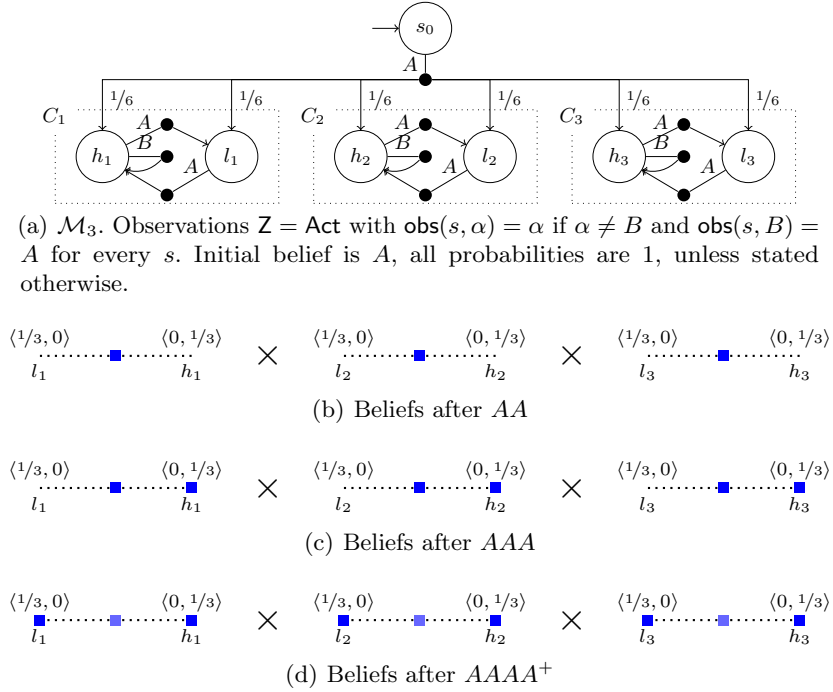
\begin{figure}[t]
\centering
\subfigure[$\mdp_3$. Observations  $\Obs = \Act$ with $\obsfun(s,\act) = \act$ if $\act \neq B$ and $\obsfun(s,B) = A$ for every $s$. Initial belief is $ A$, all probabilities are $1$, unless stated otherwise. ]{
\scalebox{0.85}{
\begin{tikzpicture}
	\node[state, initial, initial text=] at (0,2) (in) {$s_0$};
	\node[state] at (-5,0) (h1) {$h_1$};
	\node[state] at (-3,0) (l1) {$l_1$};
	 \node[state] at (-1,0) (h2) {$h_2$};
\node[state] at (1,0) (l2) {$l_2$};
\node[state] at (3,0) (h3) {$h_3$};
\node[state] at (5,0) (l3) {$l_3$};

\node[fit=(h1)(l1), inner sep=9pt, draw, dotted] (C1) {};
\node[fit=(h2)(l2), inner sep=9pt, draw, dotted] (C2) {};
\node[fit=(h3)(l3), inner sep=9pt, draw, dotted] (C3) {};

\node[fill=white] at (C1.160) (c1l) {$C_1$};
\node[fill=white] at (C2.160) (c2l) {$C_2$};
\node[fill=white] at (C3.160) (c3l) {$C_3$};

\node[circle,fill,inner sep=2pt] at (0,1.2) (astart) {};
\draw[-] (in) -- node[left,align=center] {$A$} (astart);

\draw[->] (astart) -| node[right,pos=0.7] {$\nicefrac{1}{6}$} (h1);
\draw[->] (astart) -| node[right,pos=0.7] {$\nicefrac{1}{6}$} (l1);
\draw[->] (astart) -| node[right,pos=0.7] {$\nicefrac{1}{6}$} (h2);
\draw[->] (astart) -| node[right,pos=0.7] {$\nicefrac{1}{6}$} (l2);
\draw[->] (astart) -| node[right,pos=0.7] {$\nicefrac{1}{6}$} (h3);
\draw[->] (astart) -| node[right,pos=0.7] {$\nicefrac{1}{6}$} (l3);

\node[circle,fill,inner sep=2pt] at (-4, 0.5) (h1A) {};
\node[circle,fill,inner sep=2pt] at (-4, 0) (h1B) {};
\node[circle,fill,inner sep=2pt] at (-4, -0.6) (l1A) {};
\node[circle,fill,inner sep=2pt] at (0, 0.5) (h2A) {};
\node[circle,fill,inner sep=2pt] at (0, 0) (h2B) {};
\node[circle,fill,inner sep=2pt] at (0, -0.6) (l2A) {};
\node[circle,fill,inner sep=2pt] at (4, 0.5) (h3A) {};
\node[circle,fill,inner sep=2pt] at (4, 0) (h3B) {};
\node[circle,fill,inner sep=2pt] at (4, -0.6) (l3A) {};

\draw[-] (h1) edge node[above] {$A$} (h1A);
\draw[->] (h1A) edge (l1);
\draw[-] (h1) edge node[above] {$B$} (h1B);
\draw[->] (h1B) edge[bend left] (h1);
\draw[-] (l1) edge node[above] {$A$} (l1A);
\draw[->] (l1A) edge (h1);
\draw[-] (h2) edge node[above] {$A$} (h2A);
\draw[->] (h2A) edge (l2);
\draw[-] (h2) edge node[above] {$B$} (h2B);
\draw[->] (h2B) edge[bend left] (h2);
\draw[-] (l2) edge node[above] {$A$} (l2A);
\draw[->] (l2A) edge (h2);
\draw[-] (h3) edge node[above] {$A$} (h3A);
\draw[->] (h3A) edge (l3);
\draw[-] (h3) edge node[above] {$B$} (h3B);
\draw[->] (h3B) edge[bend left] (h3);
\draw[-] (l3) edge node[above] {$A$} (l3A);
\draw[->] (l3A) edge (h3);

\end{tikzpicture}	
}
\label{fig:m3}
}
\subfigure[Beliefs after $AA$]{
\begin{tikzpicture}
	\draw[-, dotted, thick] (0,0) -- (2,0);
	\node[anchor=north] at (0,0) (l1) {\scriptsize $l_1$};
	\node[anchor=north] at (2,0) (h1) {\scriptsize $h_1$};	
	\node at (3,0) (t1) {\Large$\times$};
	\draw[-, dotted, thick] (4,0) -- (6,0);
	\node[anchor=north] at (4,0) (l2) {\scriptsize $l_2$};
	\node[anchor=north] at (6,0) (h2) {\scriptsize $h_2$};
	\node at (7,0) (t1) {\Large$\times$};
	\draw[-, dotted, thick] (8,0) -- (10,0);
	\node[anchor=north] at (8,0) (l3) {\scriptsize $l_3$};
	\node[anchor=north] at (10,0) (h3) {\scriptsize $h_3$};
	
	\node[anchor=south] at (0,0) (l1) {\scriptsize $\langle \nicefrac{1}{3}, 0\rangle$};
	\node[anchor=south] at (2,0) (l1) {\scriptsize $\langle 0, \nicefrac{1}{3}\rangle$};
	\node[anchor=south] at (4,0) (l1) {\scriptsize $\langle \nicefrac{1}{3}, 0\rangle$};
	\node[anchor=south] at (6,0) (l1) {\scriptsize $\langle 0, \nicefrac{1}{3}\rangle$};
	\node[anchor=south] at (8,0) (l1) {\scriptsize $\langle \nicefrac{1}{3}, 0\rangle$};
	\node[anchor=south] at (10,0) (l1) {\scriptsize $\langle 0, \nicefrac{1}{3}\rangle$};

	\node[fill=blue,inner sep=2pt] at (1,0) (b1) {};
	\node[fill=blue,inner sep=2pt] at (5,0) (b1) {};
	\node[fill=blue,inner sep=2pt] at (9,0) (b1) {};
	
\end{tikzpicture}
}
\subfigure[Beliefs after $AAA$]{
\begin{tikzpicture}
	\draw[-, dotted, thick] (0,0) -- (2,0);
	\node[anchor=north] at (0,0) (l1) {\scriptsize $l_1$};
	\node[anchor=north] at (2,0) (h1) {\scriptsize $h_1$};	
	\node at (3,0) (t1) {\Large$\times$};
	\draw[-, dotted, thick] (4,0) -- (6,0);
	\node[anchor=north] at (4,0) (l2) {\scriptsize $l_2$};
	\node[anchor=north] at (6,0) (h2) {\scriptsize $h_2$};
	\node at (7,0) (t1) {\Large$\times$};
	\draw[-, dotted, thick] (8,0) -- (10,0);
	\node[anchor=north] at (8,0) (l3) {\scriptsize $l_3$};
	\node[anchor=north] at (10,0) (h3) {\scriptsize $h_3$};
	
	\node[anchor=south] at (0,0) (l1) {\scriptsize $\langle \nicefrac{1}{3}, 0\rangle$};
	\node[anchor=south] at (2,0) (l1) {\scriptsize $\langle 0, \nicefrac{1}{3}\rangle$};
	\node[anchor=south] at (4,0) (l1) {\scriptsize $\langle \nicefrac{1}{3}, 0\rangle$};
	\node[anchor=south] at (6,0) (l1) {\scriptsize $\langle 0, \nicefrac{1}{3}\rangle$};
	\node[anchor=south] at (8,0) (l1) {\scriptsize $\langle \nicefrac{1}{3}, 0\rangle$};
	\node[anchor=south] at (10,0) (l1) {\scriptsize $\langle 0, \nicefrac{1}{3}\rangle$};

	\node[fill=blue,inner sep=2pt] at (1,0) (b1) {};
	\node[fill=blue,inner sep=2pt] at (5,0) (b1) {};
	\node[fill=blue,inner sep=2pt] at (9,0) (b1) {};
	\node[fill=blue,inner sep=2pt] at (2,0) (b1) {};
	\node[fill=blue,inner sep=2pt] at (6,0) (b1) {};
	\node[fill=blue,inner sep=2pt] at (10,0) (b1) {};
\end{tikzpicture}
}
\subfigure[Beliefs after $AAAA^{+}$]{
\begin{tikzpicture}
	\draw[-, dotted, thick] (0,0) -- (2,0);
	\node[anchor=north] at (0,0) (l1) {\scriptsize $l_1$};
	\node[anchor=north] at (2,0) (h1) {\scriptsize $h_1$};	
	\node at (3,0) (t1) {\Large$\times$};
	\draw[-, dotted, thick] (4,0) -- (6,0);
	\node[anchor=north] at (4,0) (l2) {\scriptsize $l_2$};
	\node[anchor=north] at (6,0) (h2) {\scriptsize $h_2$};
	\node at (7,0) (t1) {\Large$\times$};
	\draw[-, dotted, thick] (8,0) -- (10,0);
	\node[anchor=north] at (8,0) (l3) {\scriptsize $l_3$};
	\node[anchor=north] at (10,0) (h3) {\scriptsize $h_3$};
	
	\node[anchor=south] at (0,0) (l1) {\scriptsize $\langle \nicefrac{1}{3}, 0\rangle$};
	\node[anchor=south] at (2,0) (l1) {\scriptsize $\langle 0, \nicefrac{1}{3}\rangle$};
	\node[anchor=south] at (4,0) (l1) {\scriptsize $\langle \nicefrac{1}{3}, 0\rangle$};
	\node[anchor=south] at (6,0) (l1) {\scriptsize $\langle 0, \nicefrac{1}{3}\rangle$};
	\node[anchor=south] at (8,0) (l1) {\scriptsize $\langle \nicefrac{1}{3}, 0\rangle$};
	\node[anchor=south] at (10,0) (l1) {\scriptsize $\langle 0, \nicefrac{1}{3}\rangle$};

	\node[fill=blue!60,inner sep=2pt] at (1,0) (b1) {};
	\node[fill=blue!60,inner sep=2pt] at (5,0) (b1) {};
	\node[fill=blue!60,inner sep=2pt] at (9,0) (b1) {};
	\node[fill=blue,inner sep=2pt] at (2,0) (b1) {};
	\node[fill=blue,inner sep=2pt] at (6,0) (b1) {};
	\node[fill=blue,inner sep=2pt] at (10,0) (b1) {};
	
	\node[fill=blue,inner sep=2pt] at (0,0) (b1) {};
	\node[fill=blue,inner sep=2pt] at (4,0) (b1) {};
	\node[fill=blue,inner sep=2pt] at (8,0) (b1) {};
\end{tikzpicture}
}
\caption{Construction for the correctness of Lemma~\ref{lem:exponentialbelief}.}
\label{fig:m3beliefs}
\end{figure}
\noindent\emph{Proof sketch.}
We construct $\mdp_n$ with $n=3$, that is, $\mdp_3$ in Fig.~\ref{fig:m3}. 
For this MDP and $\observationtrace = AAA$, $|\vertices{\estimatorND(\trace)}| = 2^3$. 
In particular, observe how the belief factorizes into a belief within each component $C_i =  \{ h_i, l_i \}$ and notice that $\mdp_n$ has components $C_1$ to $C_n$. 
In particular, for each component, the belief being that we are with probability mass $\nicefrac{1}{n}$ (for $n=3, \nicefrac{1}{3}$) in the 'low' state $l_i$ or the 'high' state $h_i$. We depict the beliefs in Fig.~\ref{fig:m3beliefs}(b,c,d). 
Thus, for any $\tau$ with $|\observationtrace| > 2$ we can compactly represent $\vertices{\estimatorND(\observationtrace)}$ as bit-strings of length $n$.
Concretely, the belief 
\begin{align*}
	& \{ h_1, l_2, l_3  \mapsto \nicefrac{1}{3}, l_1, h_2, h_3 \mapsto 0 \} \text{ maps to } 100, \text{ and } \\
	& \{ h_1, l_2, h_3  \mapsto \nicefrac{1}{3}, l_1, h_2, l_3 \mapsto 0 \} \text{ maps to } 101. 
\end{align*}
These are exponentially many beliefs for bit strings of length $n$.\qed

One might ask whether a symbolic encoding of an exponentially large set may result in a more tractable approach to filtering.
While \Cref{theo:chcomp} allows to compute the associated risk from a set of linear constraints with standard techniques, it is not clear whether the concise set of constraints can be efficiently constructed and updated in every step. We leave this concern for future work.

In the remainder we investigate whether we need to track all these beliefs. First, when the monitor is unaware of the state-risk, this is trivially unavoidable. 
More precisely, all vertices may induce the maximal weighted trace risk by choosing an appropriate state-risk: 

\begin{lemma}
For every $\observationtrace$ and every $\belief \in \vertices{\estimatorND(\observationtrace)}$  there exists an $\staterisk$ s.t.\ 
\[ \sum_{s \in S} \belief  (s) \cdot \staterisk(s) \geq \max \limits_{\belief' \in \vertices{\estimatorND(\observationtrace)} \setminus \{ \belief \}} \sum_{s \in S} \belief'  (s) \cdot \staterisk(s) \text{ with $\max_{\belief \in \emptyset} = -\infty.$ }  \]
\end{lemma}
\noindent\emph{Proof sketch.}
We construct $r$ such that $r(s) > r(s')$ if $\belief(s) > \belief(s')$.
\qed

Second, even if the monitor is aware of the state risk $\staterisk$, it may not be able to prune enough vertices to avoid exponential growth. The crux here is that while some of the current beliefs may induce a smaller risk, an extension of the trace may cause the belief to evolve into a belief that induces the maximal risk. 
\begin{theorem}
\label{thm:extensionthm}
There exist MDPs $\mdp_n$ a $\trace$ with $B \colonequals \vertices{\estimatorND(\observationtrace)}$ and a state-risk $r$ such that
 $|B|=2^n$  and for all $\belief \in B$ exists $\observationtrace' \in \Obs^{+}$ with 	$ \tracerisk{\staterisk}(\observationtrace \cdot \observationtrace') > \sup_{\belief \in B'} \sum_s \belief(s) \cdot \staterisk(s)$, 
	where $B' = \estimatorNDupdate(B \setminus \{ \belief \}, \observationtrace' )$.  
\end{theorem}
It is helpful to understand this theorem as describing the outcome of a game between monitor and environment:
The statement says if the monitor decides to drop some vertices from $\estimatorND{\trace}$, the environment may produce an observation trace $\observationtrace'$ that will lead the monitor to underestimate the weighted risk at $\tracerisk{\staterisk}(\observationtrace \cdot \observationtrace')$.

\smallskip\noindent\emph{Proof sketch.}
We extend the construction of Fig.~\ref{fig:m3} with choices to go to a final state. The full proof sketch can be found in \Cref{sec:exponentialnecessaryproof}. 

\subsection{Approximation by pruning}
\noindent Finally, we illustrate that we cannot simply prune small probabilities from beliefs. This indicates that an approximative version of filtering for the monitoring problem is nontrivial.   
Reconsider observing $\textcolor{blue}{\obs_0\obs_0}$ in the MDP of Fig.~\ref{fig:expbeliefs}, and, for the sake of argument, let us prune the (small) entry $s_3 \mapsto \nicefrac{1}{4}$ to $0$. 
Now, continuing with the trace $\textcolor{blue}{\obs_0\obs_0}\textcolor{orange}{\obs_1}$, we would update the beliefs from before and then conclude that this trace cannot be observed with positive probability. 
With pruning, there is no upper bound on the difference between the \emph{computed} $\tracerisk{\observationtrace}$  and the \emph{actual} $\tracerisk{\observationtrace}$.
Thus, forward filtering is, in general, not tractable on MDPs. 

\section{Unrolling with Model Checking}
\label{sec:unrolling}
We present a tractable algorithm for the monitoring problem. 
Contrary to filtering, this method incorporates the state risk. 	We briefly consider the qualitative case.
An algorithm that solves that problem iteratively guesses a successor such that the given trace has positive probability, and reaches a state with sufficient risk. The algorithm only stores the current and next state and a counter. 
\begin{theorem}
\label{thm:realizabilitycomplexity}
	The Monitoring Problem with $\lambda=0$ is in NLOGSPACE.
	\end{theorem}
This result implies the existence of a polynomial time algorithm, e.g., using a graph-search on a graph growing in $|\trace|$.
There also is a deterministic algorithm with space complexity $\mathcal{O}(log^2( |\system| + |\tau| ))$, which follows from applying Savitch's Theorem~\cite{DBLP:journals/jcss/Savitch70}
, but that algorithm has exponential time complexity.

We now present a tractable algorithm for the quantitative case, where we need to store all paths.
 We do this efficiently by storing an unrolled MDP with these paths using ideas from~\cite{DBLP:conf/tacas/BaierKKM14,DBLP:journals/ai/ChatterjeeCGK16}. 
 In particular, on this MDP, we can efficiently obtain the scheduler that optimizes the risk by model checking rather than enumerating over all schedulers explicitly. 
  We give the result before going into details.
\begin{theorem}
	The Monitoring Problem (with $\lambda > 0$) is P-complete.
\end{theorem}
	The problem is P-hard, as unary-encoded step-bounded reachability is P-hard~\cite{DBLP:journals/mor/PapadimitriouT87}. 
It remains to show a P-time algorithm\footnote{On first sight, this might be surprising as step-bounded reachability in MDPs is PSPACE-hard and only quasi-polynomial. However, our problem gets a trace and therefore (assuming that the trace is not compressed) can be handled in time polynomial in the length of the trace.}, which is outlined below. Roughly, the algorithm constructs an MDP $\mdp'''$ from $\system$ in three conceptual steps, such that the maximal probability of reaching a state in $\mdp'''$ coincides with the $\tracerisk{\staterisk}(\observationtrace)$. The former can be solved by linear programming in polynomial time. The downside is that even in the best case, the memory consumption grows linearly in $|\tau|$.

We outline the main steps of the algorithm and exemplify them below:
First, we transform $\mdp$ into an MDP $\mdp'$ with \emph{deterministic state} observations, i.e., with $\obsfun'\colon S \rightarrow \Obs$. This construction is detailed in \cite[Remark 1]{DBLP:journals/ai/ChatterjeeCGK16}, and runs in polynomial time. The new initial distribution takes into account the initial observation and the initial distribution. Importantly, for each path $\path$ and each trace $\tau$, $\obsfuntrace(\path)(\trace)$ is preserved.
From here, the idea for the algorithm is a tailored adaption of the construction for conditional reachability probabilities in~\cite{DBLP:conf/tacas/BaierKKM14}. 
We ensure that $\staterisk(s) \in [0,1]$ by scaling $\staterisk$ and $\lambda$ accordingly.
Now, we construct a new MDP $\mdp'' = \langle S'', \init'', \Act'', P'' \rangle$
with state space $S'' \colonequals (S' \times \{0,\hdots ,|\trace|{-}1\}) \cup \{ \bot, \top \}$ and an $n$-times unrolled transition relation.
Furthermore, from the states $\langle s, |\trace|{-}1 \rangle$, there is a single outgoing action that with probability $\staterisk(s)$ leads to $\top$ and with probability $1-\staterisk(s)$ leads to $\bot$. Observe that the risk is now the supremum of conditioned reachability probabilities over paths that reach $\top$, conditioned by the trace $\observationtrace$. The MDP $\mdp''$ is only polynomially larger.
Then, we construct MDP $\mdp'''$ by copying $\mdp''$ and replacing (part of) the transition relation $P''$ by $P'''$ such that paths $\path$ with $\trace \not\in \obsfuntrace(\path)$ are looped back to the initial state (resembling rejection sampling). Formally, 
\begin{align*}
 P'''(\langle s, i \rangle, \act)=  \begin{cases}    P''(\langle s, i \rangle, \act) & \text{if }\obsfun'(s) = \trace_{i},\\
   \init & \text{otherwise}. 
 \end{cases}
\end{align*}
The maximal conditional reachability probability in $\mdp''$ is the maximal reachability probability in $M'''$~\cite{DBLP:conf/tacas/BaierKKM14}. Maximal reachability probabilities can be computed by solving a linear program~\cite{Put94}, and can thus be computed in polynomial time.

\begin{figure}[t]
\centering
    \subfigure[$\system$]{
    \begin{tikzpicture}[every node/.style={font=\scriptsize},simpstate/.style={circle,draw,inner sep=1pt,minimum width=14pt},act/.style={circle,fill=black,inner sep=1pt},obs/.style={circle,minimum width=10pt},baseline=(s0)]
		\node[simpstate,fill=blue!30] (s0) {$s_0$};
		\node[simpstate,below=0.8cm of s0] (s1) {$s_1$};
		\coordinate (c2) at (s1.center);
		\filldraw[draw=white,fill=blue!30] (s1.center) -- ($(c2) + (0:7pt)$) arc[start angle=0,end angle=90,radius=7pt] -- cycle;
		\filldraw[draw=white,fill=orange!30] (s1.center) -- ($(c2) + (90:7pt)$) arc[start angle=90,end angle=360,radius=7pt] -- cycle;
		\node[simpstate,below=0.8cm of s0] (s1) {$s_1$};

		\node[act,above=0.3cm of s0] (a00) {}; 
		\node[act,right=0.3cm of s0] (a01) {}; 
		\node[obs, right=0.1cm of a01] (o2) {};		
		\coordinate (c3) at (o2.center);
		\node[act,above=0.3cm of s1] (a10) {}; 
		\node[obs, left=0.1cm of a10,yshift=0.15cm] (o3) {};	
		\coordinate (c3) at (o3.center);
%
		
		\draw[-] (s0) edge node[right] {$\alpha$} (a00);
		\draw[-] (s0) edge node[above] {$\beta$} (a01);
		\draw[-] (s1) edge node[right] {$\alpha$} (a10);
		
		\draw[->] (a00) edge[bend right] node[left] {$1$} (s0);
		\draw[->] (a01) edge[bend left]  node[right] {$1$} (s1);
		\draw[->] (a10) edge node[right] {$\nicefrac{1}{2}$} (s0);
		\draw[->] (a10) edge[bend right] node[left] {$\nicefrac{1}{2}$} (s1);

		\node[simpstate,draw=white,below=1.1cm of s1] (sph) {\phantom{$s_2$}};
	\end{tikzpicture}
	\label{fig:poltimeconstruction:orig}
    }
	\subfigure[$\mdp'$]{
	\begin{tikzpicture}[every node/.style={font=\scriptsize},simpstate/.style={circle,draw,fill=white,inner sep=1pt,minimum width=14pt},act/.style={circle,fill=black,inner sep=1pt},baseline=(s0)]
		\node[simpstate,fill=blue!30] (s0) {$s_0$};
		\node[simpstate,below=0.8cm of s0,fill=blue!30] (s1) {$s_1$};
		
		\node[simpstate,draw=white,below=1.1cm of s1] (sph) {\phantom{$s_2$}};
		\node[simpstate,below=0.8cm of s1,fill=orange!30] (s2) {$s_2$};

			\node[act,above=0.3cm of s0] (a00) {}; 
		\node[act,right=0.31cm of s0] (a01) {}; 
		\node[act,above=0.3cm of s1] (a10) {}; 
		\node[act,left=0.33cm of s2] (a20) {}; 
	
		\draw[-] (s0) edge node {} (a00);
		\draw[-] (s0) edge node {} (a01);
		\draw[-] (s1) edge node {} (a10);
		\draw[-] (s2) edge node {} (a20);
		
		\draw[->] (a00) edge[bend right] node[left] {$1$} (s0);
		\draw[->] (a01) |- node[below,pos=0.75] {$\nicefrac{1}{4}$} (s1);
		\draw[->] (a01) |- node[below,pos=0.7] {$\nicefrac{3}{4}$} (s2);
		\draw[->] (a10) edge node[left] {$\nicefrac{1}{2}$} (s0);
		\draw[->] (a10) edge[bend left] node[right] {$\nicefrac{1}{8}$} (s1);
		\draw[->] (a10) edge[bend right=40] node[pos=0.1, left] {$\nicefrac{3}{8}$} (s2.100);
		\draw[->] (a20) |- node[left,pos=0.1] {$\nicefrac{1}{2}$} (s0);
		\draw[->] (a20) edge[bend right=15] node[left,pos=0.55] {$\nicefrac{1}{8}$} (s1.270);
		\draw[->] (a20) edge[bend right] node[below] {$\nicefrac{3}{8}$} (s2);

	\end{tikzpicture}
	\label{fig:poltimeconstruction:detobs}
	}
	\subfigure[$\mdp''$]{
	\begin{tikzpicture}[every node/.style={font=\scriptsize},simpstate/.style={circle,draw,fill=white,inner sep=1pt,minimum width=14pt},baseline=(s00),act/.style={circle,fill=black,inner sep=1pt}]
		\node[simpstate,fill=blue!30] (s00) {$s_0^0$};
		\node[simpstate,below=0.8cm of s00,fill=blue!30] (s10) {$s_1^0$};
		
		\node[simpstate,draw=white,below=1.1cm of s10] (sph) {\phantom{$s_2$}};
		\node[simpstate,below=0.8cm of s10,fill=orange!30] (s20) {$s_2^0$};
		\node[simpstate,right=of s00,fill=blue!30] (s01) {$s_0^1$};
		\node[simpstate,below=0.8cm of s01,fill=blue!30] (s11) {$s_1^1$};
		\node[simpstate,below=0.8cm of s11,fill=orange!30] (s21) {$s_2^1$};
		\node[simpstate,right=0.4cm of s01] (bot) {$\bot$};
		\node[simpstate,below=0.8cm of bot] (top) {$\top$};
		
			\node[act,above=0.3cm of s00] (a000) {}; 
		\node[act,right=0.2cm of s00] (a010) {}; 
		\node[act,right=0.2cm of s10] (a100) {}; 
		\node[act,right=0.2cm of s20] (a200) {}; 
		
		\draw[-] (s00) edge node {} (a000);
		\draw[-] (s00) edge node {} (a010);
		\draw[-] (s10) edge node {} (a100);
		\draw[-] (s20) edge node {} (a200);
		\draw[->] (a000) edge[bend left] node[below] {$1$} (s01);
		\draw[->] (a010) edge node[pos=0.1,right] {$\nicefrac{1}{4}$} (s11);
		\draw[->] (a010) edge[bend left=10] node[pos=0.1,left] {$\nicefrac{3}{4}$} (s21);
		\draw[->] (a100) edge[bend left=13] node[pos=0.3,left] {$\nicefrac{1}{2}$} (s01);
		\draw[->] (a100) edge node[pos=0.4,above] {$\nicefrac{1}{8}$} (s11);
		\draw[->] (a100) edge[bend right=13] node[pos=0.2,left] {$\nicefrac{3}{8}$} (s21);
		\draw[->] (a200) edge[bend right=5] node[pos=0.14,left] {$\nicefrac{1}{2}$} (s01);
		\draw[->] (a200) edge[bend right=5] node[pos=0.1,right] {$\nicefrac{1}{8}$} (s11);
		\draw[->] (a200) edge node[pos=0.1,below] {$\nicefrac{3}{8}$} (s21);
		
		\node[act,right=0.1cm of s01] (a011) {}; 
		\node[act,right=0.1cm of s11] (a101) {}; 
		\node[act,right=0.1cm of s21] (a201) {}; 
		\draw[-] (s01) edge node {} (a011);
		\draw[-] (s11) edge node {} (a101);
		\draw[-] (s21) edge node {} (a201);
		
		\draw[->] (a011) edge node[right] {$\nicefrac{1}{2}$} (top);
		\draw[->] (a101) edge node[above] {$1$} (top);
		\draw[->] (a201) edge node[left] {$1$} (top);
		\draw[->] (a011) edge node[above] {$\nicefrac{1}{2}$} (bot);

	\end{tikzpicture}
		\label{fig:poltimeconstruction:unf}
	}
	\subfigure[$\mdp'''$]{
	\begin{tikzpicture}[every node/.style={font=\scriptsize},simpstate/.style={circle,draw,fill=white,inner sep=1pt,minimum width=14pt},baseline=(s00),act/.style={circle,fill=black,inner sep=1pt}]
		\node[simpstate,fill=blue!30] (s00) {$s_0^0$};
		\node[simpstate,below=0.8cm of s00,fill=blue!30] (s10) {$s_1^0$};
		
		\node[simpstate,draw=white,below=1.1cm of s10] (sph) {\phantom{$s_2$}};
		\node[simpstate,below=0.6cm of s10,fill=orange!30] (s20) {$s_2^0$};
		\node[simpstate,right=1cm of s00,fill=blue!30] (s01) {$s_0^1$};
		\node[simpstate,below=0.8cm of s01,fill=blue!30] (s11) {$s_1^1$};
		\node[simpstate,below=0.6cm of s11,fill=orange!30] (s21) {$s_2^1$};
		\node[simpstate,right=0.4cm of s01] (bot) {$\bot$};
		\node[simpstate,below=0.8cm of bot] (top) {$\top$};
		
			\node[act,above=0.3cm of s00] (a000) {}; 
		\node[act,right=0.3cm of s00] (a010) {}; 
		\node[act,right=0.3cm of s10] (a100) {}; 
		\node[act,below=0.1cm of s20] (a200) {}; 
		\node[act,right=0.1cm of s01] (a011) {}; 
		\node[act,right=0.1cm of s11] (a101) {}; 
		\node[act,below=0.14cm of s21] (a201) {}; 
		\draw[-] (s01) edge node {} (a011);
		\draw[-] (s11) edge node {} (a101);
		\draw[-] (s21) edge node {} (a201);
		
		\draw[-] (s00) edge node {} (a000);
		\draw[-] (s00) edge node {} (a010);
		\draw[-] (s10) edge node {} (a100);
		\draw[-] (s20) edge node {} (a200);
		\draw[->] (a000) edge[bend left] node[below] {$1$} (s01);
		\draw[->] (a010) edge node[pos=0.1,right] {$\nicefrac{1}{4}$} (s11);
		\draw[->] (a010) edge[bend left=10] node[pos=0.1,left] {$\nicefrac{3}{4}$} (s21);
		\draw[->] (a100) edge[bend left=13] node[pos=0.3,left] {$\nicefrac{1}{2}$} (s01);
		\draw[->] (a100) edge node[pos=0.4,above] {$\nicefrac{1}{8}$} (s11);
		\draw[->] (a100) edge[bend right=13] node[pos=0.2,left] {$\nicefrac{3}{8}$} (s21);
				\draw[->] (a200) -- +(-0.5,0) |- node[left] {$\nicefrac{1}{2}$} (s00);
		\draw[->] (a200) -- +(-0.5,0) |- node[left] {$\nicefrac{1}{2}$} (s10);
		\draw[->] (a201) -- +(-2.05,0) |- node[left] {} (s00);
		\draw[->] (a201) -- +(-2.05,0) |- node[left] {} (s10);

		\draw[->] (a011) edge node[right] {$\nicefrac{1}{2}$} (top);
		\draw[->] (a101) edge node[above] {$1$} (top);
		\draw[->] (a011) edge node[above] {$\nicefrac{1}{2}$} (bot);
		
	\end{tikzpicture}
			\label{fig:poltimeconstruction:uncond}

	}
	\caption{Polynomial-time algorithm for solving Problem~1 illustrated.}
	\label{fig:polconstruction}
\end{figure}
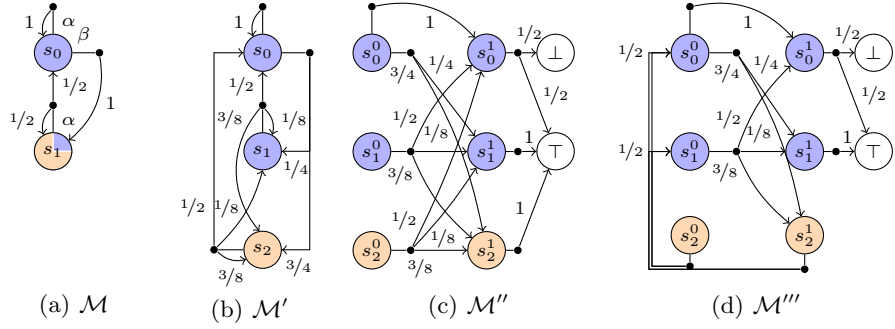
\begin{example}
We illustrate the construction in Fig.~\ref{fig:polconstruction}.
In Fig.~\ref{fig:poltimeconstruction:orig}, we depict an MDP $\system$, with $\iota = \{ s_0, s_1 \mapsto \nicefrac{1}{2} \}$.
Furthermore, let $\observationtrace = \textcolor{blue!100}{\obs_0\obs_0}$ and let $\staterisk(s_0) = 1$ and $\staterisk(s_1) = 2$.
 Let $\obsfun(s_0) = \{ \textcolor{blue}{\obs_0} \mapsto 1 \}$ and $\obsfun(s_1) = \{ \textcolor{blue!100}{\obs_0} \mapsto \nicefrac{1}{4}, \textcolor{orange!100}{\obs_1} \mapsto \nicefrac{3}{4} \}$.
State $s_1$ has two possible observations, so we split $s_1$ into $s_1$ and $s_2$ in MDP $\mdp'$, each with their own observations. Any transition into $s_1$ is now split.
As $|\observationtrace| = 2$, we unroll the MDP $\mdp'$ into MDP $\mdp''$ to represent two steps, and add goal and sink states. After rescaling, we obtain that $r(s_0) = \nicefrac{1}{2}$, whereas $r(s_1) = r(s_2) = \nicefrac{2}{2} =  1$, and we add the appropriate outgoing transitions to the states $s^1_\ast$.
In a final step, we create MDP $\mdp'''$ from $\mdp''$: we reroute all probability mass that does not agree with the observations to the initial states. Now, $\tracerisk{\staterisk}(\textcolor{blue!100}{\obs_0\obs_0})$ is  given by the probability to reach, in $\mdp'''$, in an unbounded number of steps, $\top$.
\end{example}
The construction also implies that maximizing over a finite set of schedulers, namely the deterministic schedulers with a counter from $0$ to $|\tau|$, suffices.
We denote this class $\dcscheds(|\tau|)$. 
Formally, a scheduler is in $\dcscheds(k)$ if $\text{for all } \pi,\pi'$: \[  \Big(\last{\pi} = \last{\pi'} \land \big(|\pi| = |\pi'| \lor (|\pi| > k \land |\pi'| > k)\big)\Big) \text{ implies } \sched(\pi) = \sched(\pi').\]
\begin{lemma}For every $\observationtrace$, it holds that
\label{lem:dcschedssuffice}
\[
	\tracerisk{\staterisk}(\observationtrace) \quad=\quad \max_{\sched \in \color{red}\dcscheds(|\trace|)\color{black}} \sum_{\path \in \boundedpaths{M}{|\trace|}} \prpath[\sched]{\path \mid \trace} \cdot \staterisk(\last{\path}).
	\]
\end{lemma}
The crucial idea underpinning this lemma is that memoryless schedulers suffice for the unrolling, and that the states of the unrolling can be uniquely mapped to a state and the length of the history for every $\pi$ through $\mdp$.
By reducing step-bounded reachability we can also show that this set of schedulers is necessary~\cite{DBLP:conf/formats/AndovaHK03}.

%
%

\section{Empirical Evaluation}
\label{sec:empirical}
\paragraph{Implementation.}

We provide prototype implementations for both filtering- and model-checking-based approaches from Sec.~\ref{sec:algorithms}, built on top of the probabilistic model checker \storm~\cite{DBLP:journals/corr/abs-2002-07080}. We provide a schematic setup of our implementation in Fig.~\ref{fig:prototype}.
\begin{figure}[t]
\centering
	\begin{tikzpicture}
		\node[rectangle, draw, fill=black!4, align=center,inner sep=2pt] (sim) {System\\Simulator};
		\node[rectangle, draw, fill=black!4, inner sep=2pt, align=center,above=8mm of sim, xshift=-1.7cm] (mc) {Model\\checker};
		\node[rectangle,left=4mm of mc,align=center] (tp) {State risk\\\scriptsize{temporal spec}};
		\node[rectangle,below=0.2cm of tp,align=center] (pm) {MDP\\\scriptsize{(Prism lang.)}};
		\node[below=2mm of pm] (seed) {Seed};

		\node[right=12mm of sim,draw=blue,fill=blue!10,inner sep=2pt,align=center, minimum width=21mm] (unf) {\textbf{Unrolling}\\Sec.~\ref{sec:unrolling}};
		\node[above=2mm of unf,draw=blue,fill=blue!10,align=center, inner sep=2pt, minimum width=21mm] (ff) {\textbf{Forw.~Filter}\\Sec.~\ref{sec:pruning}};

		\node[right=8mm of ff,align=center] (ch) {\scriptsize{Convex}\\\scriptsize{Hull}};
		\node[right=8mm of unf,align=center] (mc2) {\scriptsize{Model}\\\scriptsize{checking}};
		
		\node[fit=(ff)(unf)(mc2)(ch),draw=green!50!black,dashed,label={[label distance=1mm]south:\color{green!50!black}Monitor},fill=green!4] (mon) {};

		\node[right=12mm of sim,draw=blue,fill=blue!10,inner sep=2pt,align=center, minimum width=21mm] (unf) {\textbf{Unrolling}\\Sec.~\ref{sec:unrolling}};
		\node[above=2mm of unf,draw=blue,fill=blue!10,align=center, inner sep=2pt, minimum width=21mm] (ff) {\textbf{Forw.~Filter}\\Sec.~\ref{sec:pruning}};

		\node[right=8mm of ff,align=center] (ch) {\scriptsize{Convex}\\\scriptsize{Hull}};
		\node[right=8mm of unf,align=center] (mc2) {\scriptsize{Model}\\\scriptsize{checking}};

		\draw[->] (mc.0) -- node[above] {$r\colon S\rightarrow \RRnn$} +(2.4,0) -- (mon);
		\draw[->] (sim.0) -- node[above] {obs $\obs_i$} (mon.192);
		
		\node[rectangle, draw, above=2mm of mon] (mem) {\scriptsize{$<$internal memory$>$ after $\obs_0 \hdots \obs_{i-1}$}};
		
		\draw[->] (mon.70) -- (mon.70|-mem.south);
		\draw[<-] (mon.110) -- (mon.110|-mem.south);

		\node[fit=(sim)(ff)(unf)(mc2)(ch)(mon)(mem),draw,dotted,label={north:iterative process},inner sep=4pt] (proc) {};
		\node[fit=(seed)(pm)(tp),draw,dotted,label={north:input},inner sep=4pt] (inp) {};
		
		\draw[->,dotted] (ff) -- node[above] {\scriptsize{\emph{uses}}} (ch);
		\draw[->,dotted] (unf) -- node[above] {\scriptsize{\emph{uses}}} (mc2);
		
		\draw[->] (seed) -- (sim);
		\draw[->] (tp) -- (mc);
		
		\draw[->] (pm.0) -- (mc);
		
		\draw[->] (pm.0) -- (sim);
		\draw[->] (pm.0) -- (mon);
		\node[right=5mm of mon,align=center] (risk) {risk\\$R_i$};
		\draw[->] (mon.0) -- (risk);

	\end{tikzpicture}
	\caption{Schematic setup for prototype mapping stream $\obs_0\hdots \obs_k$ to stream $R_0 \hdots R_k$.}
	\label{fig:prototype}
\end{figure}
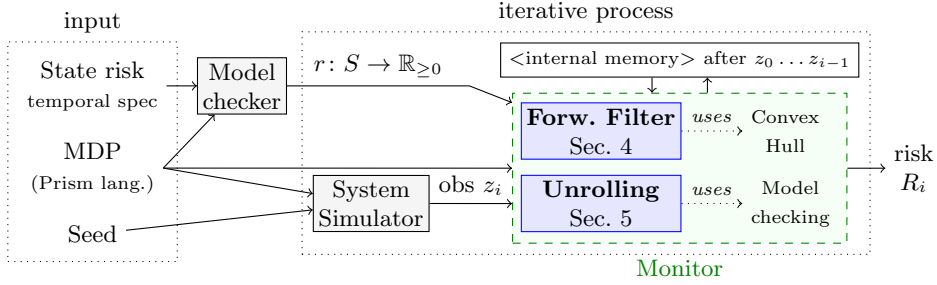
As input, we consider a symbolic description of MDPs with state-based observation labels, based on an extended dialect of the Prism language.
We define the state risk in this MDP via a temporal property (given as a \textsc{PCTL} formula), and obtain the concrete state-risk by model checking.
We take a seed that yields a trace using the simulator. For the experiments, actions are resolved uniformly in this simulator\footnote{This is not an assumption but rather our evaluation strategy.}.
The simulator iteratively feeds observations into the monitor, running either of our two algorithms (implemented in C++). After each observation $z_i$, the monitor computes the risk $R_i$ having observed $\obs_0\hdots\obs_i$. 
We flexibly combine these components via a Python API~\footnote{Available at \url{https://github.com/monitoring-MDPs/premise}}.

For filtering as in Sec.~\ref{sec:pruning}, we provide a sparse data structure for beliefs that is updated using only deterministic schedulers. This is sufficient, see \Cref{lem:convexsuffices}.
To further prune the set of beliefs, we implement an SMT-driven elimination~\cite{DBLP:reference/cg/Seidel04} of interior beliefs, inside of the convex hull\footnote{Advanced algorithms like Quickhull~\cite{DBLP:journals/toms/BarberDH96} are not without significant adaptions applicable as the set of beliefs can be degenerate  (roughly, a set without full rank).}.
We construct the unrolling as described in Sec.~\ref{sec:unrolling} and apply model checking via any sparse engines in \storm. 

\paragraph*{Reproducibility.}
We archived a container with sources, benchmarks, and scripts to reproduce our experiments: {\url{https://doi.org/10.5281/zenodo.4724622}}.

\paragraph*{Set-up.}
For each benchmark described below, we sampled 50 random traces using seeds 0--49 of lengths up to $|\trace|=500$. 
We are interested in the \emph{promptness}, that is, the delay of 
 time between getting an observation $\obs_i$ and returning corresponding risk $r_i$, as well as the \emph{cumulative performance} obtained by summing over the promptness along the trace. We use a timeout of $1$ second for this query. 
We compare the forward filtering (FF) approach with and without convex hull (CH) reduction, and the model unrolling approach (UNR) with two model checking engines of \storm: exact policy iteration (EPI, \cite{Put94}) and optimistic value iteration~(OVI, \cite{DBLP:conf/cav/HartmannsK20}).
All experiments are run on a MacBook Pro MV962LL/A, using a single core. The memory limit of 6GB was not violated. We use Z3~\cite{DBLP:conf/tacas/MouraB08} as SMT-solver~\cite{DBLP:series/faia/BarrettSST09} for the convex hull reduction.

\paragraph*{Benchmarks.} We present three benchmark families, all MDPs with a combination of probabilities, nondeterminism and partial observability. 

\smallskip\noindent\benchmark{Airport-A} is as in Sec.~\ref{sec:examples}, but with a higher resolution for both ground vehicle in the middle lane and the plane. \benchmark{Airport-B} has a two-state sensor model with stochastic transitions between them.

\smallskip\noindent\benchmark{Refuel-A}  models robots with a depleting battery and recharging stations. The world model consists of a robot moving around in a  $D{\times}D$ grid with some dedicated charging cells, where each action costs energy. The risk is to deplete the battery within a fixed horizon. \benchmark{Refuel-B} is a two-state sensor variant.

\smallskip\noindent\benchmark{Evade-I} is inspired by a navigation task in a multi-agent setting in a   $D{\times}D$ grid. The monitored robot moves randomly, and the risk is defined as the probability of crashing with the other robot. 
The other robot has an internal incentive in the form of a cardinal direction, and nondeterministically decides to move or to uniformly randomly change its incentive.
The monitor observes everything except the incentive of the other robot. 
\benchmark{Evade-V} is an alternative navigation task: Contrary to above, the other robot does not have an internal state and indeed navigates nondeterministically in one of the cardinal directions. We only observe the other robot  location is within the view range.

	\begin{table}[t]
	\centering
	\caption{Performance for promptness of online monitoring on various benchmarks.}
	{\footnotesize
	\scalebox{0.86}{

     \begin{tabular}{lll|rrr|rrrrrrr|rrrrr|}
    &      &      &          &               &         & \multicolumn{7}{c|}{Forward Filtering} & \multicolumn{5}{c|}{Unrolling}  \\
    Id & Name & Inst & $|S|$ & $|P|$ & $|\trace|$ & $N$ &  $\underset{\text{avg}}{T}$ & $\underset{\text{max}}{T}$ &  $\underset{\text{avg}}{B}$ & $\underset{\text{max}}{B}$     & $\underset{\text{avg}}{D}$ & $\underset{\text{max}}{D}$ & $N$ &  $\underset{\text{avg}}{T}$ & $\underset{\text{max}}{T}$ & $\underset{\text{avg}}{|S_u|}$ & $\underset{\text{max}}{|S_u|}$ \\\hline
    
\multirow{2}{*}{1}	& \multirow{2}{*}{\benchmark{airport-A}}	& \multirow{2}{*}{7,50,30}	& \multirow{2}{*}{20910}	& \multirow{2}{*}{114143} & 100	&\highlightPASS{50}	& 0.01	& 0.01	& 4.5	& 7	& 4.6	& 7	& \highlightPASS{50}	& 0.04	& 0.11	& 524	& 599\\
			&	&	&	&	&500	&\highlightPASS{50}	& 0.01	& 0.01	& 1.0	& 1	& 1.0	& 1	& \highlightPASS{50}	& 0.01	& 0.01	& 1075	& 1258\\
\hline
\multirow{2}{*}{2}	& \multirow{2}{*}{\benchmark{airport-B}}	& \multirow{2}{*}{3,50,30}	& \multirow{2}{*}{20232}	& \multirow{2}{*}{106012} & 100	&\highlightFAIL{0}	& 	& 	& 	& 	& 	& 	& \highlightPASS{50} 	& 0.09	& 0.16	& 556	& 629\\
			&	&	&	&	&500	&\highlightFAIL{0}	& 	& 	& 	& 	& 	& 	& \highlightPASS{50}	& 0.01	& 0.01	& 1460	& 1647\\
\hline
\multirow{2}{*}{3}	& \multirow{2}{*}{\benchmark{airport-B}}	& \multirow{2}{*}{7,50,30}	& \multirow{2}{*}{41820}	& \multirow{2}{*}{308474} & 100	&\highlightFAIL{0}	& 	& 	& 	& 	& 	& 	& \highlightPASS{50}	& 0.14	& 0.33	& 1000	& 1183\\
			&	&	&	&	&500	&\highlightFAIL{0}	& 	& 	& 	& 	& 	& 	& \highlightFAIL{11} 	& 0.02	& 0.02	& 2097	& 2297\\
\hline
\multirow{2}{*}{4}	& \multirow{2}{*}{\benchmark{refuel-A}}	& \multirow{2}{*}{12,50}	& \multirow{2}{*}{45073}	& \multirow{2}{*}{2431691} & 100	&\highlightPASS{50}	& 0.01	& 0.01	& 2.2	& 4	& 2.8	& 5	& \highlightPASS{50}	& 0.01	& 0.05	& 325	& 409\\
			&	&	&	&	&500	&\highlightPASS{50}	& 0.01	& 0.01	& 1.5	& 4	& 1.7	& 5	& \highlightPASS{50} 	& 0.01	& 0.19	& 1071	& 2409\\
\hline
\multirow{2}{*}{5}	& \multirow{2}{*}{\benchmark{refuel-B}}	& \multirow{2}{*}{12,50}	& \multirow{2}{*}{90145}	& \multirow{2}{*}{9725277} & 100	&\highlightPASS{50}	& 0.06	& 0.23	& 4.2	& 8	& 5.6	& 10	& \highlightPASS{50}	& 0.04	& 0.17	& 608	& 732\\
			&	&	&	&	&500	&\highlightPASS{50}	& 0.01	& 0.01	& 2.9	& 8	& 3.3	& 10	& \highlightFAIL{46} 	& 0.04	& 0.09	& 2171	& 4688\\
\hline
\multirow{2}{*}{6}	& \multirow{2}{*}{\benchmark{evade-I}}	& \multirow{2}{*}{15}	& \multirow{2}{*}{377101}	& \multirow{2}{*}{2022295} & 100	&\highlightPASS{50}	& 0.01	& 0.02	& 2.6	& 10	& 3.3	& 4	& \highlightFAIL{49}	& 0.01	& 0.06	& 332	& 363\\
			&	&	&	&	&500	&\highlightPASS{50}	& 0.01	& 0.01	& 2.4	& 5	& 3.4	& 4	& \highlightFAIL{45}	& 0.08	& 0.90	& 1655	& 1891\\
\hline
\multirow{2}{*}{7}	& \multirow{2}{*}{\benchmark{evade-V}}	& \multirow{2}{*}{5,3}	& \multirow{2}{*}{1001}	& \multirow{2}{*}{5318} & 100	&\highlightFAIL{26}	& 0.01	& 0.01	& 1.0	& 1	& 1.0	& 1	& \highlightPASS{50}	& 0.00	& 0.02	& 134	& 241\\
			&	&	&	&	&500	&\highlightFAIL{25}	& 0.01	& 0.01	& 1.0	& 1	& 1.0	& 1	& \highlightPASS{50}	& 0.00	& 0.01	& 538	& 671\\
\hline
\multirow{2}{*}{8}	& \multirow{2}{*}{\benchmark{evade-V}}	& \multirow{2}{*}{6,3}	& \multirow{2}{*}{2161}	& \multirow{2}{*}{11817} & 100	&\highlightFAIL{1}	& 0.01	& 0.01	& 1.0	& 1	& 1.0	& 1	& \highlightPASS{50}	& 0.02	& 0.32	& 319	& 861\\
			&	&	&	&	&500	&\highlightFAIL{1}	& 0.01	& 0.01	& 1.0	& 1	& 1.0	& 1	& \highlightFAIL{49}	& 0.01	& 0.02	& 777	& 1484\\
\hline

    \end{tabular}
    
}}
\label{tab:promptness}
\vspace{-0.2cm}
\end{table}

\paragraph*{Results.}
We split our results in two tables.
In Table~\ref{tab:promptness}, we give an ID for every benchmark name and instance, along with the size of the MDP (nr.\ of states $|S|$ and transitions $|P|$) our algorithms operate on. 
We consider the promptness after prefixes of length $|\trace|$.
In particular, for forward filtering with the convex hull optimization, we give the number $N$ of traces that did not time out before, and consider the average $T_\text{avg}$ and maximal time $T_\text{max}$ needed (over all sampled traces that did not time-out before). 
Furthermore, we give the average, $B_\text{avg}$, and maximal, $B_\text{max}$, number of beliefs stored (after reduction), and the average, $D_\text{avg}$, and maximal, $D_\text{max}$,  dimension of the belief support. 
Likewise, for unrolling with exact model checking, we give the number $N$ of traces that did not time out before, and we consider average $T_\text{avg}$ and maximal time $T_\text{max}$, as well as the average size and maximal number of states of the unfolded MDP.
\begin{table}
\caption{Summarized performance for online monitoring}
\centering
{\footnotesize
\scalebox{0.86}{

          \begin{tabular}{lr|rrrrr|rrrr|rrrrr|rrr|}
             &      &  \multicolumn{5}{c|}{FF w/ CH}  &  \multicolumn{4}{c|}{FF w/o CH}      & \multicolumn{5}{c|}{UNR (EPI)} & \multicolumn{3}{c|}{UNR (OVI)}   \\
        Id & $|\trace|$ & $N$ & $\underset{\text{avg}}{T}$ & $\underset{\text{max}}{T}$ & $\underset{\text{avg}}{B}$ & $\underset{\text{avg}}{E}$ & $N$ & $\underset{\text{avg}}{T}$ & $\underset{\text{max}}{T}$ & $\underset{\text{avg}}{B}$  & $N$ & $\underset{\text{avg}}{T}$ & $\underset{\text{max}}{T}$ &    $\underset{\text{avg}}{Bld^\%}$ & $\underset{\text{max}}{Bld^\%}$ & N  & $\underset{\text{avg}}{T}$ & $\underset{\text{max}}{T}$  \\\hline
        
1	& 100	& \highlightPASS{50}	& 0.9	& 1.1	& 493	& 241	& \highlightFAIL{0}	& 	& 	& 	& \highlightPASS{50}	& 2.9	& 3.6	& 6	& 56	& \highlightPASS{50}	& 0.0	& 0.1\\
			&	500	& \highlightPASS{50}	& 3.7	& 4.3	& 1040	& 316	& \highlightFAIL{0}	& 	& 	& 	& \highlightPASS{50}	& 7.5	& 10.7	& 21	& 24	& \highlightPASS{50}	& 0.4	& 0.8\\
\hline
2	& 100	& \highlightFAIL{0}	& 	& 	& 	& 	& \highlightFAIL{0}	& 	& 	& 	& \highlightPASS{50}	& 3.7	& 4.7	& 6	& 54	& \highlightPASS{50}	& 0.1	& 0.1\\
			&	500	& \highlightFAIL{0}	& 	& 	& 	& 	& \highlightFAIL{0}	& 	& 	& 	& \highlightPASS{50}	& 11.9	& 17.1	& 18	& 23	& \highlightPASS{50}	& 0.6	& 0.8\\
\hline
3	& 100	& \highlightFAIL{0}	& 	& 	& 	& 	& \highlightFAIL{0}	& 	& 	& 	& \highlightPASS{50}	& 7.6	& 10.6	& 5	& 55	& \highlightPASS{50}	& 0.1	& 0.2\\
			&	500	& \highlightFAIL{0}	& 	& 	& 	& 	& \highlightFAIL{0}	& 	& 	& 	& \highlightFAIL{11}	& 21.3	& 28.7	& 19	& 23	& \highlightPASS{50}	& 0.9	& 1.7\\
\hline
4	& 100	& \highlightPASS{50}	& 0.7	& 0.8	& 241	& 138	& \highlightFAIL{1}	& 0.9	& 0.9	& 1473	& \highlightPASS{50}	& 0.7	& 1.0	& 35	& 69	& \highlightPASS{50}	& 0.0	& 0.1\\
			&	500	& \highlightPASS{50}	& 3.4	& 3.7	& 868	& 226	& \highlightFAIL{1}	& 0.9	& 0.9	& 1873	& \highlightPASS{50}	& 5.6	& 21.2	& 57	& 67	& \highlightPASS{50}	& 0.5	& 0.9\\
\hline
5	& 100	& \highlightPASS{50}	& 7.4	& 10.7	& 442	& 2267	& \highlightFAIL{0}	& 	& 	& 	& \highlightPASS{50}	& 2.5	& 4.4	& 32	& 57	& \highlightPASS{50}	& 0.1	& 0.2\\
			&	500	& \highlightPASS{50}	& 16.5	& 42.2	& 1781	& 4249	& \highlightFAIL{0}	& 	& 	& 	& \highlightFAIL{46}	& 19.5	& 64.2	& 55	& 70	& \highlightPASS{50}	& 1.3	& 2.3\\
\hline
6	& 100	& \highlightPASS{50}	& 1.1	& 4.8	& 273	& 160	& \highlightFAIL{13}	& 0.7	& 2.9	& 2055	& \highlightFAIL{49}	& 0.5	& 2.0	& 34	& 65	& \highlightFAIL{47}	& 0.0	& 0.1\\
			&	500	& \highlightPASS{50}	& 5.1	& 11.5	& 1237	& 632	& \highlightFAIL{2}	& 4.4	& 6.8	& 20524	& \highlightFAIL{45}	& 22.4	& 53.6	& 13	& 29	& \highlightFAIL{43}	& 0.5	& 0.7\\
\hline
7	& 100	& \highlightFAIL{26}	& 0.8	& 1.2	& 106	& 11	& \highlightFAIL{13}	& 0.1	& 0.5	& 274	& \highlightPASS{50}	& 0.4	& 1.0	& 19	& 45	& \highlightFAIL{48}	& 0.0	& 0.1\\
			&	500	& \highlightFAIL{25}	& 3.7	& 4.2	& 505	& 7	& \highlightFAIL{13}	& 0.1	& 0.5	& 674	& \highlightPASS{50}	& 1.3	& 4.4	& 46	& 58	& \highlightFAIL{47}	& 0.2	& 0.3\\
\hline
8	& 100	& \highlightFAIL{1}	& 1.3	& 1.3	& 124	& 109	& \highlightFAIL{0}	& 	& 	& 	& \highlightPASS{50}	& 1.5	& 7.0	& 15	& 39	& \highlightFAIL{36}	& 0.4	& 5.6\\
			&	500	& \highlightFAIL{1}	& 4.3	& 4.3	& 524	& 109	& \highlightFAIL{0}	& 	& 	& 	& \highlightFAIL{49}	& 4.9	& 28.1	& 37	& 56	& \highlightFAIL{35}	& 0.7	& 6.4\\
\hline

        \end{tabular}
        	
}}
\label{tab:performance}
\vspace{-0.2cm}
\end{table}
In Table~\ref{tab:performance}, we consider for the benchmarks above the cumulative performance. In particular, this table also considers an alternative implementation for both FF and UNR.
We use the IDs to identify the instance, and sum for each prefix of length $|\trace|$ the time.
For filtering, we recall the number of traces $N$ that did not time out, the average and maximal cumulative time along the trace, the average cumulative number of beliefs that were considered, and the average cumulative number of beliefs eliminated. For the case without convex hull, we do not eliminate any vertices.  
For unrolling, we report average $T_\text{avg}$ and maximal cumulative time using EPI, as well as the time required for model building, $\mathit{Bld^{\%}}$ (relative to the total time, per trace). 
We compare this to the average and maximal cumulative time for using OVI (notice that building times remain approximately the same).

\paragraph*{Discussion.}
The results from our prototype show that conservative (sound) predictive modeling of systems that combine probabilities, nondeterminism and partial observability is within reach with the methods we proposed and state-of-the-art algorithms. 
Both forward filtering and an unrolling-based approaches have their merits. The practical results thus slightly diverge from the complexity results in Sec.~\ref{sec:operationalmodel}, due to structural properties of some benchmarks. In particular, for \benchmark{airport-A} and \benchmark{refuel-A}, the nondeterminism barely influences the belief, and so there is no explosion, and consequentially the dimension of the belief is sufficiently small that the convex hull can be efficiently computed. 
Rather than the number of states, this  belief dimension makes \benchmark{evade-V} a difficult benchmark\footnote{The max dimension ${=}1$ in \benchmark{evade-V} is only over the traces that did not time-out. The dimension when running in time-outs is above  $5$.}. \emph{If many states can be reached with a particular trace, and if along these paths there are some probabilistic states, forward filtering suffers significantly}. We see that if the benchmark allows for efficacious forward filtering, it is not slowed down in the way that unrolling is slower on longer traces. 
For UNR, we observe that OVI is typically the fastest, but EPI does not suffer from the numerical worst-cases as OVI does. \emph{If an observation trace is unlikely, the unrolled MDP constitutes a numerically challenging problem, in particular for  value-iteration  based model checkers}, see~\cite{DBLP:journals/tcs/HaddadM18}.
For FF, the convex hull computation is essential for any dimension, and eliminating some vertices in every step keeps the number of belief states manageable. 
%


\section{Related work}
We are not the first to consider model-based runtime verification in the presence of partial observability and probabilities.
Runtime verification with state estimation on hidden Markov models (HMM)---without nondeterminism has been studied for various types of properties~\cite{DBLP:conf/vmcai/SistlaS08,RTVStochasticFaultySystems,DBLP:conf/rv/StollerBSGHSZ11} and has  been extended to hybrid systems~\cite{RTMStochasticCPS}.
The tool Prevent focusses on black-box systems by learning an HMM from a set of traces. The HMM approximates (with only convergence-in-the-limit guarantees) the actual system~\cite{DBLP:conf/sefm/BabaeeGF18},
and then estimates during runtime the most likely trace rather than estimating a  distribution over current states. Extensions consider symmetry reductions on the models~\cite{DBLP:conf/rv/BabaeeGF18}.
These techniques do not make a conservative (sound) risk estimation.
The recent framework for runtime verification in the presence of partial observability~\cite{DBLP:conf/rv/CimattiTT19} takes a more strict black-box view and cannot provide state estimates.
Finally, \cite{DBLP:conf/concur/GrigoreK18} chooses to have partial observability to make monitoring of software systems more efficient, and \cite{DBLP:conf/rtss/WooM07} monitors a noisy sensor to reduce  energy consumption.

State beliefs are studied when verifying HMMs~\cite{DBLP:conf/forte/ZhangHJ05}, where the question whether a sequence of observations likely occurs, or which HMM is an adequate representation of a system~\cite{DBLP:conf/lics/KieferS16}.
State beliefs are prominent in the verification of partially observable MDPs~\cite{DBLP:journals/rts/Norman0Z17,DBLP:conf/ijcai/HorakBC18,DBLP:conf/atva/BorkJKQ20}, where one can observe the actions taken (but the problem itself is to find the right scheduler). 
Our monitoring problem can be phrased as a special case of verification of partially observable stochastic games~\cite{DBLP:journals/tocl/Chatterjee014}, but automatic techniques for those very general models are  lacking.
Likewise, the idea of \emph{shielding} (pre)computes all action choices that lead to safe behavior~\cite{DBLP:conf/tacas/BloemKKW15,DBLP:conf/aaai/AlshiekhBEKNT18,DBLP:journals/corr/DragerFK0U15,DBLP:conf/tacas/Junges0DTK16,DBLP:conf/cav/AvniBCHKP19,DBLP:conf/concur/0001KJSB20}. 
For partially observable settings, shielding again requires to compute partial-information schedulers~\cite{DBLP:journals/tsmc/NamA10,DBLP:conf/aaai/Chatterjee0PRZ17}, contrary to our approach.
Partial observability has also been  studied in the context of diagnosability, studying if a fault has occurred (in the past)~\cite{DBLP:journals/iandc/BertrandHL19}, or what actions uncover faults~\cite{DBLP:conf/fossacs/BertrandFHHH14}.
 We, instead assume partial observability in which we do detect faults, but want to estimate the risk that these faults occur in the future.

The assurance framework for reinforcement learning~\cite{DBLP:journals/corr/abs-1908-00528} implicitly allows for stochastic behavior, but cannot cope with partial observability or nondeterminism.
Predictive monitoring has been combined with deep learning~\cite{DBLP:conf/rv/BortolussiCPSS19} and Bayesian inference~\cite{SriramIros2020}, where the key problem is that the computation of an imminent failure is too expensive to be done exactly.
 More generally, learning automata models has been motivated with runtime assurance~\cite{DBLP:conf/pts/AichernigB0HPRR19,DBLP:conf/fm/TapplerA0EL19}. 
Testing approaches statistically evaluate whether traces are likely to be produced by a given model~\cite{DBLP:journals/fac/GerholdS18}. 
The approach in \cite{DBLP:journals/fmsd/AichernigT19} studies stochastic black-box systems with controllable nondeterminism and iteratively learns a model for the system.
\section{Conclusion}
We have presented the first framework for monitoring based on a trace of observations on models that combine nondeterminism and probabilities.
Future work includes heuristics for approximate monitoring and  for faster convex hull computations, and to apply this work to grey-box (learned) models. 
\bibliographystyle{splncs04}
\bibliography{literature.bib}

\begin{thebibliography}{10}
\providecommand{\url}[1]{\texttt{#1}}
\providecommand{\urlprefix}{URL }
\providecommand{\doi}[1]{https://doi.org/#1}

\bibitem{DBLP:conf/pts/AichernigB0HPRR19}
Aichernig, B.K., Bloem, R., Ebrahimi, M., Horn, M., Pernkopf, F., Roth, W.,
  Rupp, A., Tappler, M., Tranninger, M.: Learning a behavior model of hybrid
  systems through combining model-based testing and machine learning. In:
  {ICTSS}. {LNCS}, vol. 11812, pp. 3--21. Springer (2019)

\bibitem{DBLP:journals/fmsd/AichernigT19}
Aichernig, B.K., Tappler, M.: Probabilistic black-box reachability checking
  (extended version). Formal Methods Syst. Des.  \textbf{54}(3),  416--448
  (2019)

\bibitem{DBLP:conf/aaai/AlshiekhBEKNT18}
Alshiekh, M., Bloem, R., Ehlers, R., K{\"{o}}nighofer, B., Niekum, S., Topcu,
  U.: Safe reinforcement learning via shielding. In: {AAAI}. pp. 2669--2678.
  {AAAI} Press (2018)

\bibitem{DBLP:conf/formats/AndovaHK03}
Andova, S., Hermanns, H., Katoen, J.: Discrete-time rewards model-checked. In:
  {FORMATS}. {LNCS}, vol.~2791, pp. 88--104. Springer (2003)

\bibitem{DBLP:conf/cav/AvniBCHKP19}
Avni, G., Bloem, R., Chatterjee, K., Henzinger, T.A., K{\"{o}}nighofer, B.,
  Pranger, S.: Run-time optimization for learned controllers through
  quantitative games. In: {CAV} {(1)}. {LNCS}, vol. 11561, pp. 630--649.
  Springer (2019)

\bibitem{DBLP:conf/sefm/BabaeeGF18}
Babaee, R., Gurfinkel, A., Fischmeister, S.: \emph{P}revent : {A} predictive
  run-time verification framework using statistical learning. In: {SEFM}.
  {LNCS}, vol. 10886, pp. 205--220. Springer (2018)

\bibitem{DBLP:conf/rv/BabaeeGF18}
Babaee, R., Gurfinkel, A., Fischmeister, S.: Predictive run-time verification
  of discrete-time reachability properties in black-box systems using
  trace-level abstraction and statistical learning. In: {RV}. {LNCS}, vol.
  11237, pp. 187--204. Springer (2018)

\bibitem{BK08}
Baier, C., Katoen, J.: Principles of model checking. {MIT} Press (2008)

\bibitem{DBLP:conf/tacas/BaierKKM14}
Baier, C., Klein, J., Kl{\"{u}}ppelholz, S., M{\"{a}}rcker, S.: Computing
  conditional probabilities in {M}arkovian models efficiently. In: {TACAS}.
  {LNCS}, vol.~8413, pp. 515--530. Springer (2014)

\bibitem{DBLP:journals/toms/BarberDH96}
Barber, C.B., Dobkin, D.P., Huhdanpaa, H.: The quickhull algorithm for convex
  hulls. {ACM} Trans. Math. Softw.  \textbf{22}(4),  469--483 (1996)

\bibitem{DBLP:series/faia/BarrettSST09}
Barrett, C.W., Sebastiani, R., Seshia, S.A., Tinelli, C.: Satisfiability modulo
  theories. In: Handbook of Satisfiability, Frontiers in Artificial
  Intelligence and Applications, vol.~185, pp. 825--885. {IOS} Press (2009)

\bibitem{DBLP:series/lncs/BartocciDDFMNS18}
Bartocci, E., Deshmukh, J.V., Donz{\'{e}}, A., Fainekos, G.E., Maler, O.,
  Nickovic, D., Sankaranarayanan, S.: Specification-based monitoring of
  cyber-physical systems: {A} survey on theory, tools and applications. In:
  Lectures on Runtime Verification - Introductory and Advanced Topics, {LNCS},
  vol. 10457, pp. 135--175. Springer (2018)

\bibitem{DBLP:conf/fossacs/BertrandFHHH14}
Bertrand, N., Fabre, E., Haar, S., Haddad, S., H{\'{e}}lou{\"{e}}t, L.: Active
  diagnosis for probabilistic systems. In: FoSSaCS. {LNCS}, vol.~8412, pp.
  29--42. Springer (2014)

\bibitem{DBLP:journals/iandc/BertrandHL19}
Bertrand, N., Haddad, S., Lefaucheux, E.: A tale of two diagnoses in
  probabilistic systems. Inf. Comput.  \textbf{269} (2019)

\bibitem{DBLP:conf/tacas/BloemKKW15}
Bloem, R., K{\"{o}}nighofer, B., K{\"{o}}nighofer, R., Wang, C.: Shield
  synthesis: - runtime enforcement for reactive systems. In: {TACAS}. {LNCS},
  vol.~9035, pp. 533--548. Springer (2015)

\bibitem{DBLP:conf/atva/BorkJKQ20}
Bork, A., Junges, S., Katoen, J., Quatmann, T.: Verification of
  indefinite-horizon {POMDP}s. In: {ATVA}. {LNCS}, vol. {12302}, pp. 288--304.
  Springer (2020)

\bibitem{DBLP:conf/rv/BortolussiCPSS19}
Bortolussi, L., Cairoli, F., Paoletti, N., Smolka, S.A., Stoller, S.D.: Neural
  predictive monitoring. In: {RV}. {LNCS}, vol. 11757, pp. 129--147. Springer
  (2019)

\bibitem{DBLP:books/cu/BV2014}
Boyd, S.P., Vandenberghe, L.: Convex Optimization. Cambridge University Press
  (2014)

\bibitem{DBLP:journals/ai/ChatterjeeCGK16}
Chatterjee, K., Chmelik, M., Gupta, R., Kanodia, A.: Optimal cost almost-sure
  reachability in pomdps. Artif. Intell.  \textbf{234},  26--48 (2016)

\bibitem{DBLP:journals/tocl/Chatterjee014}
Chatterjee, K., Doyen, L.: Partial-observation stochastic games: How to win
  when belief fails. {ACM} Trans. Comput. Log.  \textbf{15}(2),  16:1--16:44
  (2014)

\bibitem{DBLP:conf/aaai/Chatterjee0PRZ17}
Chatterjee, K., Novotn{\'{y}}, P., P{\'{e}}rez, G.A., Raskin, J., Zikelic, D.:
  Optimizing expectation with guarantees in {POMDP}s. In: {AAAI}. pp.
  3725--3732. {AAAI} Press (2017)

\bibitem{SriramIros2020}
Chou, Y., Yoon, H., Sankaranarayanan, S.: Predictive runtime monitoring of
  vehicle models using bayesian estimation and reachability analysis. In:
  {IROS} (2020), to appear

\bibitem{DBLP:conf/rv/CimattiTT19}
Cimatti, A., Tian, C., Tonetta, S.: Assumption-based runtime verification with
  partial observability and resets. In: {RV}. {LNCS}, vol. 11757, pp. 165--184.
  Springer (2019)

\bibitem{DBLP:journals/corr/DragerFK0U15}
Dr{\"{a}}ger, K., Forejt, V., Kwiatkowska, M.Z., Parker, D., Ujma, M.:
  Permissive controller synthesis for probabilistic systems. Log. Methods
  Comput. Sci.  \textbf{11}(2) (2015)

\bibitem{DBLP:journals/fac/GerholdS18}
Gerhold, M., Stoelinga, M.: Model-based testing of probabilistic systems.
  Formal Asp. Comput.  \textbf{30}(1),  77--106 (2018)

\bibitem{DBLP:conf/concur/GrigoreK18}
Grigore, R., Kiefer, S.: Selective monitoring. In: {CONCUR}. LIPIcs, vol.~118,
  pp. 20:1--20:16. {Dagstuhl - LZI} (2018)

\bibitem{DBLP:journals/tcs/HaddadM18}
Haddad, S., Monmege, B.: Interval iteration algorithm for {MDP}s and {IMDP}s.
  Theor. Comput. Sci.  \textbf{735},  111--131 (2018)

\bibitem{DBLP:conf/cav/HartmannsK20}
Hartmanns, A., Kaminski, B.L.: Optimistic value iteration. In: {CAV} {(2)}.
  {LNCS}, vol. 12225, pp. 488--511. Springer (2020)

\bibitem{DBLP:conf/rv/HavelundR18}
Havelund, K., Rosu, G.: Runtime verification - 17 years later. In: {RV}.
  {LNCS}, vol. 11237, pp. 3--17. Springer (2018)

\bibitem{DBLP:journals/corr/abs-2002-07080}
Hensel, C., Junges, S., Katoen, J., Quatmann, T., Volk, M.: The probabilistic
  model checker storm. CoRR  \textbf{abs/2002.07080} (2020)

\bibitem{forwardrechability}
{Henzinger}, T.A., {Pei-Hsin Ho}, {Wong-Toi}, H.: Algorithmic analysis of
  nonlinear hybrid systems. IEEE Trans. Autom. Control  \textbf{43}(4),
  540--554 (1998)

\bibitem{DBLP:conf/ijcai/HorakBC18}
Hor{\'{a}}k, K., Bosansk{\'{y}}, B., Chatterjee, K.: Goal-{HSVI}: Heuristic
  search value iteration for goal {POMDP}s. In: {IJCAI}. pp. 4764--4770.
  ijcai.org (2018)

\bibitem{DBLP:conf/nfm/0001HTT19}
Jansen, N., Humphrey, L.R., Tumova, J., Topcu, U.: Structured synthesis for
  probabilistic systems. In: {NFM}. {LNCS}, vol. 11460, pp. 237--254. Springer
  (2019)

\bibitem{DBLP:conf/concur/0001KJSB20}
Jansen, N., K{\"{o}}nighofer, B., Junges, S., Serban, A., Bloem, R.: Safe
  reinforcement learning using probabilistic shields (invited paper). In:
  {CONCUR}. LIPIcs, vol.~171, pp. 3:1--3:16. {Dagstuhl - LZI} (2020)

\bibitem{DBLP:conf/tacas/Junges0DTK16}
Junges, S., Jansen, N., Dehnert, C., Topcu, U., Katoen, J.: Safety-constrained
  reinforcement learning for {MDP}s. In: {TACAS}. {LNCS}, vol.~9636, pp.
  130--146. Springer (2016)

\bibitem{DBLP:conf/lics/KieferS16}
Kiefer, S., Sistla, A.P.: Distinguishing hidden {M}arkov chains. In: {LICS}.
  pp. 66--75. {ACM} (2016)

\bibitem{DBLP:conf/tacas/MouraB08}
de~Moura, L.M., Bj{\o}rner, N.: {Z3:} an efficient {SMT} solver. In: {TACAS}.
  {LNCS}, vol.~4963, pp. 337--340. Springer (2008)

\bibitem{DBLP:journals/tsmc/NamA10}
Nam, W., Alur, R.: Active learning of plans for safety and reachability goals
  with partial observability. {IEEE} Trans. Syst. Man Cybern. Part {B}
  \textbf{40}(2),  412--420 (2010)

\bibitem{DBLP:journals/rts/Norman0Z17}
Norman, G., Parker, D., Zou, X.: Verification and control of partially
  observable probabilistic systems. Real Time Syst.  \textbf{53}(3),  354--402
  (2017)

\bibitem{DBLP:journals/mor/PapadimitriouT87}
Papadimitriou, C.H., Tsitsiklis, J.N.: The complexity of {M}arkov decision
  processes. Math. Oper. Res.  \textbf{12}(3),  441--450 (1987)

\bibitem{DBLP:journals/corr/abs-1908-00528}
Phan, D., Paoletti, N., Grosu, R., Jansen, N., Smolka, S.A., Stoller, S.D.:
  Neural simplex architecture. CoRR  \textbf{abs/1908.00528} (2019), appears at
  NFM2020

\bibitem{Put94}
Puterman, M.L.: {M}arkov Decision Processes: Discrete Stochastic Dynamic
  Programming. Wiley Series in Probability and Statistics, Wiley (1994)

\bibitem{forwardfiltering}
{Rabiner}, L.R.: A tutorial on hidden markov models and selected applications
  in speech recognition. Proc.\ of the IEEE  \textbf{77}(2),  257--286 (1989)

\bibitem{DBLP:journals/fmsd/SanchezSABBCFFK19}
S{\'{a}}nchez, C., Schneider, G., Ahrendt, W., Bartocci, E., Bianculli, D.,
  Colombo, C., Falcone, Y., Francalanza, A., Krstic, S., Louren{\c{c}}o, J.M.,
  Nickovic, D., Pace, G.J., Rufino, J., Signoles, J., Traytel, D., Weiss, A.: A
  survey of challenges for runtime verification from advanced application
  domains (beyond software). Formal Methods Syst. Des.  \textbf{54}(3),
  279--335 (2019)

\bibitem{DBLP:journals/jcss/Savitch70}
Savitch, W.J.: Relationships between nondeterministic and deterministic tape
  complexities. J. Comput. Syst. Sci.  \textbf{4}(2),  177--192 (1970)

\bibitem{DBLP:books/daglib/0090562}
Schrijver, A.: Theory of linear and integer programming. Wiley-Interscience
  series in discrete mathematics and optimization, Wiley (1999)

\bibitem{DBLP:reference/cg/Seidel04}
Seidel, R.: Convex hull computations. In: Handbook of Discrete and
  Computational Geometry, 2nd Ed, pp. 495--512. Chapman and Hall/CRC (2004)

\bibitem{DBLP:conf/rv/Seshia19}
Seshia, S.A.: Introspective environment modeling. In: {RV}. {LNCS}, vol. 11757,
  pp. 15--26. Springer (2019)

\bibitem{DBLP:conf/vmcai/SistlaS08}
Sistla, A.P., Srinivas, A.R.: Monitoring temporal properties of stochastic
  systems. In: {VMCAI}. {LNCS}, vol.~4905, pp. 294--308. Springer (2008)

\bibitem{RTMStochasticCPS}
Sistla, A.P., {\v{Z}}efran, M., Feng, Y.: Runtime monitoring of stochastic
  cyber-physical systems with hybrid state. In: {RV}. Springer (2012)

\bibitem{DBLP:books/sp/12/Spaan12}
Spaan, M.T.J.: Partially observable {M}arkov decision processes. In:
  Reinforcement Learning, Adaptation, Learning, and Optimization, vol.~12, pp.
  387--414. Springer (2012)

\bibitem{DBLP:conf/rv/StollerBSGHSZ11}
Stoller, S.D., Bartocci, E., Seyster, J., Grosu, R., Havelund, K., Smolka,
  S.A., Zadok, E.: Runtime verification with state estimation. In: {RV}.
  {LNCS}, vol.~7186, pp. 193--207. Springer (2011)

\bibitem{DBLP:conf/fm/TapplerA0EL19}
Tappler, M., Aichernig, B.K., Bacci, G., Eichlseder, M., Larsen, K.G.:
  L\({}^{\mbox{*}}\)-based learning of {M}arkov decision processes. In: {FM}.
  {LNCS}, vol. 11800, pp. 651--669. Springer (2019)

\bibitem{DBLP:books/daglib/0014221}
Thrun, S., Burgard, W., Fox, D.: Probabilistic robotics. Intelligent robotics
  and autonomous agents, {MIT} Press (2005)

\bibitem{RTVStochasticFaultySystems}
Wilcox, C.M., Williams, B.C.: Runtime verification of stochastic, faulty
  systems. In: {RV}. {LNCS}, vol.~6418, pp. 452--459. Springer (2010)

\bibitem{DBLP:conf/rtss/WooM07}
Woo, H., Mok, A.K.: Real-time monitoring of uncertain data streams using
  probabilistic similarity. In: {RTSS}. pp. 288--300. {IEEE} {CS} (2007)

\bibitem{DBLP:conf/forte/ZhangHJ05}
Zhang, L., Hermanns, H., Jansen, D.N.: Logic and model checking for hidden
  {M}arkov models. In: {FORTE}. {LNCS}, vol.~3731, pp. 98--112. Springer (2005)

\end{thebibliography}
\clearpage
\pagebreak
\appendix

\section{On qualitative variants of the problem}
We may reformulate the qualitative problem as follows: First, define the risk as the maximum over reached states, i.e., 
\[\tracerisk{\staterisk}^{\max}(\trace) = \max \limits_{\pi  \in \finpaths{} \text{ s.t. } \obsfuntrace(\pi)(\trace)>0 } r(\last{\pi}) .\]
 This \emph{maximum risk estimation} conservatively states the highest risk that can be achieved after an observation. The quantitative (qualitative) \emph{maximal risk estimation problem} analogously is to decide $\tracerisk{\staterisk}^{\max}(\trace) > \lambda$ with $\lambda \geq 0$ ($\lambda=0$).
\begin{lemma}
\label{lem:maxrisk}
The qualitative risk estimation problem and the quantitative and qualitative maximum risk estimation problem are all logspace interreducible.
\end{lemma}

\begin{proof}[Lemma~\ref{lem:qualonkripke} ]
The qualitative case is a special case of the quantitative case.
Now consider the quantitative case, and assume some fixed $\lambda > 0$. 
Observe that for any fixed scheduler, $ \sum_{\path} \prpath[\sched]{\path \mid \trace} =1$, and that there is a unique path $\path$ such that $\prpath[\sched]{\path \mid \trace} =1$. 
We can reduce the quantitative maximum risk estimation problem to a qualitative monitoring problem as follows. We modify the state risk function by checking state-wise whether the risk at a state is above $\lambda$. If it is, then we set the state risk to $1$, and to zero otherwise. If $\tracerisk{\staterisk}(\trace)>\lambda$ for some $\lambda \in \RRnn$, then there is a path $\path \in \boundedpaths{\system}{|\trace|}$ with $\staterisk(\last{\path})>\lambda$. This in turn means that  $\prpath[\sched]{\path \mid \trace} \cdot \staterisk(\last{\path})>0$. The other way around, if the sum over all paths is zero, then there is no such path. 
\end{proof}

\begin{proof}[Lemma~\ref{lem:maxrisk} ]
\begin{itemize}
\item \emph{The qualitative weighted risk estimation problem and  the qualitative  maximum risk estimation problem coincide.}

The qualitative weighted risk estimation problem and  the qualitative  maximum risk estimation problem coincide.
	In the qualitative case $\tracerisk{\staterisk}(\trace)>0$ holds when at least one path $\path$ matches the trace $\trace$ with a positive probability and $\staterisk(\last{\path})> 0$. In this case, $\tracerisk{\staterisk}^{\max}(\trace) > 0$ also holds, since $\staterisk(\last{\path})> 0$ and $\obsfuntrace(\pi)(\trace)>0$. 

When $\tracerisk{\staterisk}^{\max}(\trace) > 0$ holds then there is a path $\path$ such that $\staterisk(\last{\path})> 0$ and $\obsfuntrace(\pi)(\trace)>0$. This implies that $\prpath{\sched}(\path\mid\trace)>0$ and in turn that $\tracerisk{\staterisk}(\trace)>0$. \qed

\item \emph{The qualitative weighted risk estimation problem and the quantitative  maximum risk estimation problem are logspace interreducible.}
From the qualitative weighted risk estimation problem to the qualitative maximum risk estimation problem is trivial and follows from the last lemma since the qualitative maximum risk estimation problem is an instance of the quantitative maximum risk estimation problem. 

In the other direction, we can reduce the quantitative maximum risk estimation problem to a qualitative weighted risk estimation problem as follows. We modify the state risk function by checking state-wise whether the risk at a state is above $\lambda$. If it is, then we set the state risk to $1$, and to zero otherwise. If $\tracerisk{\staterisk}^{\max}(\trace)>\lambda$ for some $\lambda \in \RRnn$, then there is a path $\path \in \boundedpaths{\system}{|\trace|}$ with $\staterisk(\last{\path})>\lambda$. This in turn means that  $\prpath[\sched]{\path \mid \trace} \cdot \staterisk(\last{\path})>0$. Furthermore, if $\tracerisk{\staterisk}^{\max}(\trace)\leq\lambda$, then for all $\path \in \boundedpaths{\system}{|\trace|}$, it holds that $\prpath[\sched]{\path \mid \trace} \cdot \staterisk(\last{\path})=0$. 

Both reduction can be performed in logspace. In the first direction we only need to set the lambda. In the second direction we just need to iterate over the states and change the risk function accordingly. \qed
\end{itemize}
\end{proof}

\section{Proof of \Cref{lem:kripkestructest}}
We show $\tracerisk{\staterisk}(\observationtrace) = \max_{s\in \sensbelief(\trace)} \staterisk(s)$.  
	We show this by proving the following three equalities. Let $\boundedpaths{{\trsys,\sched}}{|\tau|}  = \{ \path \in \boundedpaths{\trsys}{|\trace|} \mid \prpath[\sched]{\pi} > 0 \}$ and let $\mathit{DetSched}(\trsys)$ be the set of deterministic schedulers of $\trsys$: 
	\begin{align}
		\tracerisk{\staterisk}(\observationtrace) & = \sup_{\sched \in \scheds{\trsys}} \sum_{\path \in \boundedpaths{\trsys}{|\tau|}} \prpath[\sched]{\path \mid \trace} \cdot \staterisk(\last{\path}) \\
		 & =  \sup_{\sched \in \mathit{DetSched}(\trsys)} \sum_{\path \in \boundedpaths{\trsys}{|\tau|}} \prpath[\sched]{\path \mid \trace} \cdot \staterisk(\last{\path})  \\
		 &= \max_{\sched \in \mathit{DetSched}{\trsys}} \max_{\path \in \boundedpaths{{\trsys,\sched}}{|\tau|}}   \staterisk(\last{\path}) \\
		 &= \max_{s \in \sensbelief(\trace)} \staterisk(s) 
	\end{align}

We start with $(1)=(2)$. The case $(1)\ge (2)$ is trivial since $\mathit{DetSched}(\system) \subseteq \scheds{\system}$. The case $(1)\leq (2)$ follows from as convex sums are maximized on their vertices. The case $(2) = (3)$ holds because deterministic schedulers induce unique paths. More precisely, since all distributions in a Kripke structure are Dirac, a (deterministic) scheduler $\sched$ and a trace $\trace$ defines a unique path  $\path_\sched$ in the Kripke structure. 
	 Due to the absence of probabilities, for every path $\path$, we have $\prpath[\sched]{\path \mid \trace} \in \{ 0, 1\}$, and thus we have $\boundedpaths{{\trsys,\sched}}{|\tau|} = \{ \path_\sigma \}$.
	  For $(3)=(4)$,   the path $\pi_\sigma$ has a unique last state $s = \last{\path_\sigma}$ that is obtained having observed $\observationtrace$. 
	 Let $S_{\trace} = \{ \last{\pi} \mid \exists \sched. \{ \pi \} = \boundedpaths{{\trsys,\sched}}{|\tau|}\}$. It follows that 
	 \[ \max_{\sched \in \scheds{\trsys}} \max_{\path \in \boundedpaths{{\trsys,\sched}}{|\tau|}}   \staterisk(\last{\path}) =  \max_{s \in S_\trace} \staterisk(s) . \]
	 By induction over $|\trace|$, we can show $S_{\trace} = \sensbelief(\trace)$. 

\section{Proof of Thm~\ref{thm:extensionthm}}
\label{sec:exponentialnecessaryproof}

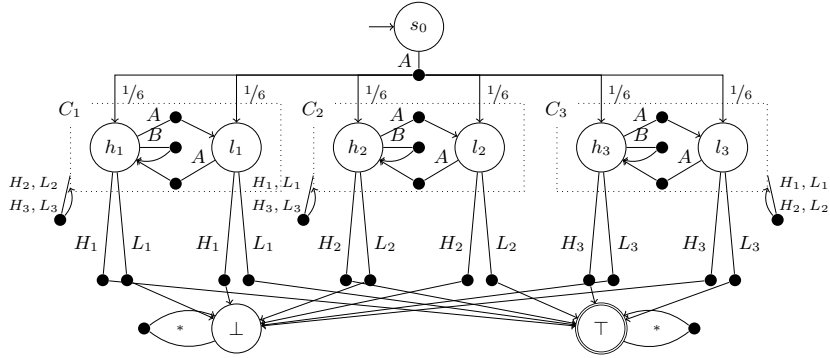
\begin{figure}[t]
\centering
\scalebox{0.8}{
\begin{tikzpicture}
	\node[state, initial, initial text=] at (0,2) (in) {$s_0$};
	\node[state] at (-5,0) (h1) {$h_1$};
	\node[state] at (-3,0) (l1) {$l_1$};
	 \node[state] at (-1,0) (h2) {$h_2$};
\node[state] at (1,0) (l2) {$l_2$};
\node[state] at (3,0) (h3) {$h_3$};
\node[state] at (5,0) (l3) {$l_3$};

	 \node[state] at (-3,-3) (bot) {$\bot$};
\node[state, accepting] at (3,-3) (top) {$\top$};

\node[fit=(h1)(l1), inner sep=9pt, draw, dotted] (C1) {};
\node[fit=(h2)(l2), inner sep=9pt, draw, dotted] (C2) {};
\node[fit=(h3)(l3), inner sep=9pt, draw, dotted] (C3) {};

\node[fill=white] at (C1.160) (c1l) {$C_1$};
\node[fill=white] at (C2.160) (c2l) {$C_2$};
\node[fill=white] at (C3.160) (c3l) {$C_3$};

\node[circle,fill,inner sep=2pt] at (-5.9,-1.2) (slf1) {};
\node[circle,fill,inner sep=2pt] at (-1.9,-1.2) (slf2) {};
\node[circle,fill,inner sep=2pt] at (5.9,-1.2) (slf3) {};

\node[circle,fill,inner sep=2pt] at (0,1.2) (astart) {};
\draw[-] (in) -- node[left,align=center] {$A$} (astart);

\draw[->] (astart) -| node[right,pos=0.7] {$\nicefrac{1}{6}$} (h1);
\draw[->] (astart) -| node[right,pos=0.7] {$\nicefrac{1}{6}$} (l1);
\draw[->] (astart) -| node[right,pos=0.7] {$\nicefrac{1}{6}$} (h2);
\draw[->] (astart) -| node[right,pos=0.7] {$\nicefrac{1}{6}$} (l2);
\draw[->] (astart) -| node[right,pos=0.7] {$\nicefrac{1}{6}$} (h3);
\draw[->] (astart) -| node[right,pos=0.7] {$\nicefrac{1}{6}$} (l3);

\node[circle,fill,inner sep=2pt] at (-4, 0.5) (h1A) {};
\node[circle,fill,inner sep=2pt] at (-4, 0) (h1B) {};
\node[circle,fill,inner sep=2pt] at (-4, -0.6) (l1A) {};
\node[circle,fill,inner sep=2pt] at (0, 0.5) (h2A) {};
\node[circle,fill,inner sep=2pt] at (0, 0) (h2B) {};
\node[circle,fill,inner sep=2pt] at (0, -0.6) (l2A) {};
\node[circle,fill,inner sep=2pt] at (4, 0.5) (h3A) {};
\node[circle,fill,inner sep=2pt] at (4, 0) (h3B) {};
\node[circle,fill,inner sep=2pt] at (4, -0.6) (l3A) {};

\draw[-] (h1) edge node[above] {$A$} (h1A);
\draw[->] (h1A) edge (l1);
\draw[-] (h1) edge node[above] {$B$} (h1B);
\draw[->] (h1B) edge[bend left] (h1);
\draw[-] (l1) edge node[above] {$A$} (l1A);
\draw[->] (l1A) edge (h1);
\draw[-] (h2) edge node[above] {$A$} (h2A);
\draw[->] (h2A) edge (l2);
\draw[-] (h2) edge node[above] {$B$} (h2B);
\draw[->] (h2B) edge[bend left] (h2);
\draw[-] (l2) edge node[above] {$A$} (l2A);
\draw[->] (l2A) edge (h2);
\draw[-] (h3) edge node[above] {$A$} (h3A);
\draw[->] (h3A) edge (l3);
\draw[-] (h3) edge node[above] {$B$} (h3B);
\draw[->] (h3B) edge[bend left] (h3);
\draw[-] (l3) edge node[above] {$A$} (l3A);
\draw[->] (l3A) edge (h3);

\node[circle,fill,inner sep=2pt] at (-5.2, -2.2) (h1H) {};
\node[circle,fill,inner sep=2pt] at (-4.8, -2.2) (h1L) {};
\node[circle,fill,inner sep=2pt] at (-3.2, -2.2) (l1H) {};
\node[circle,fill,inner sep=2pt] at (-2.8, -2.2) (l1L) {};
\node[circle,fill,inner sep=2pt] at (-1.2, -2.2) (h2H) {};
\node[circle,fill,inner sep=2pt] at (-0.8, -2.2) (h2L) {};
\node[circle,fill,inner sep=2pt] at (0.8, -2.2) (l2H) {};
\node[circle,fill,inner sep=2pt] at (1.2, -2.2) (l2L) {};
\node[circle,fill,inner sep=2pt] at (2.8, -2.2) (h3H) {};
\node[circle,fill,inner sep=2pt] at (3.2, -2.2) (h3L) {};
\node[circle,fill,inner sep=2pt] at (4.8, -2.2) (l3H) {};
\node[circle,fill,inner sep=2pt] at (5.2, -2.2) (l3L) {};

\draw[-] (h1) -- node[pos=0.7,left] {$H_1$} (h1H);
\draw[-] (h1) -- node[pos=0.7,right] {$L_1$} (h1L);
\draw[-] (l1) -- node[pos=0.7,left] {$H_1$} (l1H);
\draw[-] (l1) -- node[pos=0.7,right] {$L_1$} (l1L);
\draw[-] (h2) -- node[pos=0.7,left] {$H_2$} (h2H);
\draw[-] (h2) -- node[pos=0.7,right] {$L_2$} (h2L);
\draw[-] (l2) -- node[pos=0.7,left] {$H_2$} (l2H);
\draw[-] (l2) -- node[pos=0.7,right] {$L_2$} (l2L);
\draw[-] (h3) -- node[pos=0.7,left] {$H_3$} (h3H);
\draw[-] (h3) -- node[pos=0.7,right] {$L_3$} (h3L);
\draw[-] (l3) -- node[pos=0.7,left] {$H_3$} (l3H);
\draw[-] (l3) -- node[pos=0.7,right] {$L_3$} (l3L);

\draw[-] (C1.195) -- node[left,align=center] {\scriptsize{$H_2, L_2$}\\\scriptsize{$H_3, L_3$}} (slf1);
\draw[->] (slf1) edge[bend right] node[right] {} (C1.201);

\draw[-] (C2.195) -- node[left,align=center] {\scriptsize{$H_1, L_1$}\\\scriptsize{$H_3, L_3$}} (slf2);
\draw[->] (slf2) edge[bend right] node[right] {} (C2.201);

\draw[-] (C3.345) -- node[right,align=center] {\scriptsize{$H_1, L_1$}\\\scriptsize{$H_2, L_2$}} (slf3);
\draw[->] (slf3) edge[bend left] node[right] {} (C3.339);

\draw[->] (h1H) -- (top);
\draw[->] (h1L) -- (bot);
\draw[->] (l1H) -- (bot);
\draw[->] (l1L) -- (top);
\draw[->] (h2H)-- (top);
\draw[->] (h2L) -- (bot);
\draw[->] (l2H) -- (bot);
\draw[->] (l2L) -- (top);
\draw[->] (h3H) -- (top);
\draw[->] (h3L) -- (bot);
\draw[->] (l3H) -- (bot);
\draw[->] (l3L) -- (top);

\node[circle,fill,inner sep=2pt, right=of top] (topa) {};
\node[circle,fill,inner sep=2pt, left=of bot] (bota) {};

\draw[->] (bota) edge[bend left] (bot);
\draw[->] (topa) edge[bend right] (top);

\draw[-] (bota) edge[bend right] node[above] {$^{*}$} (bot);
\draw[-] (topa) edge[bend left] node[above] {$^{*}$} (top);

\end{tikzpicture}	
}

\caption{Construction for the correctness of Thm.~\ref{thm:extensionthm}, showing $\mdp_3$'. Observations  $\Obs = \Act$ with $\obsfun(s,\act) = \act$ if $\act \neq B$ and $\obsfun(s,B) = A$ for every $s$. Initial belief is $ A$, all probabilities are $1$, unless stated otherwise. Self-loops at the components $C_1,C_2,C_3$ reflect self-loops on the states in the components.}
\label{fig:m3ext}
\end{figure}

We extend the construction of Fig.~\ref{fig:m3}, to obtain the MDP $\mdp'_3$ in Fig.~\ref{fig:m3ext}.
By taking an action $H_i$ or $L_i$, the probability mass leaves both the low and high-state of a component $C_i$. In all other components, the probability mass remains. We can represent states with extended bit-strings,
concretely, we can think of 
\begin{align*}
& \{ h_1, l_3, \top  \mapsto \nicefrac{1}{3}, l_1, l_2, h_2, h_3, \bot \mapsto 0 \} \text{ as }1\cmark 0, \text{ and}\\
& \{  \bot \mapsto \nicefrac{2}{3}, l_3, \top  \mapsto \nicefrac{1}{3}, l_1, l_2, h_1, h_2, h_3, \bot \mapsto 0 \} \text{ as }\xmark\xmark 0.
\end{align*}
We define the risk as $r(\top) = 1$ and $r(s) = 0$ for all $s \neq \top$.
Thus, the belief $\cmark\cmark\cmark$ has maximal risk. 
Now, we show that if we have $B \colonequals \vertices{\estimatorND(AAA)} = \mathbb{B}^n$. 
Let us fix $\belief \colonequals 100$. 
Then $B' =\estimatorNDupdate(B \setminus \{ \belief \}, H_1L_2L_3)$ does not contain $\belief' \colonequals \cmark\cmark \cmark$, 
but $\vertices{\estimatorND(AAAH_1L_2L_3)}$ does. Thus, the estimated risk after $AAAH_1L_2L_3$ will be below the maximum if we remove $\belief'$ after seeing $AAA$.
Symmetrically, we can show that all other beliefs need to remain in $B$.\qed

%
%
\end{document}